\documentclass{article}
\usepackage{graphicx}
\usepackage{amsfonts}
\usepackage{amsmath}
\usepackage{amssymb}
\usepackage{url}
\usepackage{fancyhdr}
\usepackage{indentfirst}
\usepackage{enumerate}
\usepackage{algorithm}
\usepackage{algpseudocode}

\usepackage{makecell}
\usepackage{float}

\usepackage{dsfont}
\usepackage[colorlinks=true,citecolor=blue]{hyperref}
\usepackage{amsthm}
\usepackage{color}
\usepackage{natbib}
\usepackage{comment}
\usepackage{capt-of}
\usepackage{multirow}
\usepackage{subfig}
\def\d{\mathrm{d}}

\newcommand{\D}{\mathcal {D}}

\newcommand{\E}{\mathbb{E}}

\newcommand{\F}{\mathcal{F}}

\newcommand{\R}{\mathbb{R}}
\newcommand{\N}{\mathbb{N}}
\newcommand{\p}{\mathbb{P}}
\newcommand{\Q}{\mathcal{Q}}
\newcommand{\W}{\mathcal{W}}

\newcommand{\Y}{\mathcal{Y}}

\newcommand{\id}{\mathds{1}}

\renewcommand{\(}{\left(}
\renewcommand{\)}{\right)}
\renewcommand{\[}{\left[}
\renewcommand{\]}{\right]}

\renewcommand{\le}{\leqslant}
\renewcommand{\geq}{\geqslant}
\renewcommand{\leq}{\leqslant}
\renewcommand{\epsilon}{\varepsilon}

\renewcommand{\cdots}{\dots}

\theoremstyle{plain}
\newtheorem{theorem}{Theorem}

\newtheorem{lemma}{Lemma}
\newtheorem{proposition}{Proposition}
\theoremstyle{definition}
\newtheorem{definition}{Definition}
\newtheorem{example}{Example}

\newtheorem{assumption}{Assumption}
\theoremstyle{remark}
\newtheorem{remark}{Remark}

\newcommand{\cet}{\begin{center}}
	\newcommand{\ecet}{\end{center}}

\usepackage{setspace}

\begin{document}
	
\title{Stochastic Dominance Constrained Optimization with S-shaped Utilities: Poor-Performance-Region Algorithm and Neural Network} 

\author{
    Zeyun Hu\thanks{\scriptsize School of Science and Engineering, The Chinese University of Hong Kong (Shenzhen), Shenzhen, China. Email: \texttt{zeyunhu@link.cuhk.edu.cn}}
    \and
	Yang Liu\thanks{\scriptsize Corresponding author. School of Science and Engineering, The Chinese University of Hong Kong (Shenzhen), Shenzhen, China. Email: \texttt{yangliu16@cuhk.edu.cn}}
}

\date{}
\maketitle

\begin{abstract}	
	
We investigate the static portfolio selection problem of S-shaped (non-concave) utility maximization under first-order and second-order stochastic dominance (SD) constraints. In many S-shaped utility optimization problems, one should require a liquidation boundary to guarantee the existence of a finite concave envelope function. A first-order SD (FSD) constraint can replace this requirement and provide an alternative for risk management. We explicitly solve the optimal solution under a general S-shaped utility function with a first-order stochastic dominance constraint. However, the second-order SD (SSD) constrained problem under non-concave utilities is difficult to solve analytically due to the invalidity of Sion's max-min theorem. 
For this reason, we propose a numerical algorithm to obtain a plausible and sub-optimal solution for general non-concave utilities. The key idea is to detect the poor performance region with respect to the SSD constraints, characterize its structure and modify the distribution on that region to obtain (sub-)optimality. A key financial insight is that the decision maker should follow the SD constraint on the poor performance scenario while conducting the unconstrained optimal strategy otherwise. We provide numerical experiments to show that our algorithm effectively finds a sub-optimal solution in many cases. Finally, we develop an algorithm-guided piecewise-neural-network framework to learn the solution of the SSD problem, which demonstrates accelerated convergence compared to standard neural network approaches.

	

\end{abstract}

\section{Introduction}

S-shaped utility functions, formalized in cumulative prospect theory (\cite{TK1992}) and surveyed in behavioral finance (\cite{BT2003}), capture two salient features of the investor behavior: risk seeking in losses and risk aversion in gains. In portfolio selection problems, this non-concavity can induce aggressive tail-risk taking (\cite{C2000,HK2018,CHN2019,DZ2020,DKQW2022,LLZ2025}). In particular, when the pricing kernel has a heavy right tail, unconstrained S-shaped optimization typically generates a heavy left tail in terminal wealth, leading to large probabilities of extreme losses, which are the risks that classical VaR/ES-type constraints do not reliably mitigate.

In the field of optimization, stochastic dominance (SD) constraints are widely applied to portfolio selection problems with risk management (\cite{R1970,DR2003,DR2006,WX2021,WWX2024}). In this paper, we adopt stochastic dominance constraints as an implementable approach to control such downside risk under non-concave utilities. First-order SD (FSD) relative to a benchmark wealth $X_0$ enforces a quantile-by-quantile floor, offering a distribution-level safety guarantee that often aligns better with practice than hard liquidation boundaries. Second-order SD (SSD) restricts cumulative quantiles, providing a flexible, model-free form of downside protection from the pathological risk-seeking induced by the convex (loss) region of S-shaped preferences.

Analytically, SD-constrained utility maximization is well investigated for strictly concave utilities via quantile reformulations, duality, and saddle-point arguments (\cite{WX2021}). However, extending these tools to non-concave utilities under SSD is challenging. Sion’s max-min theorem no longer applies in the key saddle-point step, and the concavification principle is not guaranteed to be valid because optimizers can fall in regions where the utility and its concave envelope differ (e.g., \cite{BS2014,DKQW2022,PWW2025}), and closed-form solutions are scarce beyond special and technically constrained settings (\cite{WX2021}). This paper contributes in three aspects.

First, 
we derive an explicit optimal solution for portfolio selection with a general S-shaped utility under an FSD constraint relative to a benchmark. Crucially, the FSD constraint obviates the need for a liquidation boundary: the optimal wealth is counter-comonotone with the pricing kernel and exhibits a two-regime structure, coinciding with the classical solution in favorable states and binding to the benchmark quantile in adverse states. This clarifies how FSD serves as a “soft” and interpretable left-tail boundary.

Second, 
we propose the Poor-Performance-Region Algorithm (PPRA), a numerical method that constructs a feasible, high-quality suboptimal solution, and in many cases a numerically optimal one. The key idea is to identify the ``poor performance region", namely quantile levels where the unconstrained classical optimizer violates SSD relative to the benchmark. The algorithm partitions this region and applies a state-dependent correction to enforce SSD, while reverting to the unconstrained policy elsewhere. Financially, the resulting rule is intuitive: track the benchmark in poor states to satisfy SSD; otherwise follow the classical optimizer.

Third, 
we develop an algorithm-guided piecewise-neural-network framework that embeds the PPRA-derived partition and analytic priors into the architecture. This design drastically narrows the functional search space, accelerates convergence, satisfies budget and SSD constraints more quickly, and attains higher objective values than a standard monolithic network, especially in non-concave settings where regular training struggles with infeasibility and local minima (\cite{ST1998}).

Methodologically, our approach combines duality and concavification insights (\cite{KLS1987,C2000,LL2020}) with the quantile formulation of utility maximization (\cite{HZ2011,X2016,GZ2025}) and SD constrained optimization (\cite{WX2021,WWX2024}). For FSD, we provide a closed-form solution without a liquidation boundary and interpret FSD as a distributional safety floor. For SSD, we translate feasibility into integral inequalities in the quantile domain and design a correction that ``repairs” exactly where the unconstrained policy underperforms the benchmark.

We validate our methods in a complete-market (Black-Scholes) setting across diverse utilities (CRRA, exponential, log, S-shaped, and piecewise) and benchmarks (log-normal, normal, exponential, uniform). For FSD, the explicit solution confirms that dominance constraints can substitute for liquidation boundaries. For SSD, PPRA consistently produces feasible, interpretable solutions that often match known optima in concave cases. The piecewise neural network guided by PPRA converges substantially faster and to better solutions than a monolithic network, particularly under non-concavity. 

Financial implications are immediate. 
First, 
FSD guarantees that all terminal-wealth quantiles exceed those of a benchmark, giving a realistic and implementable floor. 
Second, 
the optimal/near-optimal policy is to adhere to the benchmark in bad states and follow the classical policy otherwise, yielding transparent risk management. Third, 
SD constraints significantly reduce the probability of extreme losses induced by S-shaped preferences, beyond what standard VaR/ES constraints typically achieve.

\textbf{Scope and limitations}. Our analysis focuses on complete markets and static terminal-wealth problems, leveraging their equivalence to dynamic continuous-time settings via the pricing kernel. While PPRA is broadly applicable and robust in experiments, it provides suboptimality guarantees rather than universal optimality in the non-concave SSD case; establishing general sharp optimality conditions remains a promising direction for future research.

The structure of this paper is as follows. Section \ref{sec:model-settings} presents the model. Section \ref{sec:FSD} derives the explicit FSD solution for general S-shaped utilities and explains how FSD replaces liquidation boundaries. Section \ref{sec:SSD} explains the analytical hurdles of SSD under non-concavity and introduces PPRA. Section \ref{sec:Numerical results} provides numerical studies across utilities and benchmarks. Section \ref{sec:NN} develops the algorithm-guided, piecewise neural-network framework and contrasts it with standard networks. Section \ref{sec:conc} concludes.

\section{Model Settings}
\label{sec:model-settings}

Fix an atomless probability space $(\Omega, \F, \p)$. Let $L^0$ be the set of all random variables on $(\Omega, \F, \p)$. Let $L^1 \subset L^0$ be the set of all integrable random variables.
Denote the pricing kernel by a continuously-distributed random variable $\rho: \Omega \to (0, \infty)$ and $\rho \in L^1$. 
For an initial capital $\overline{x}\in \R$, the static version of the classic \cite{M1969}'s problem is given by
\begin{equation}\label{prob:merton}
	\begin{aligned}
		& \max_{X \in L^0} \E\[U(X)\]	\text{ s.t. } \E[\rho X] \leq \overline{x}, 
	\end{aligned}
\end{equation}
where $U: \R \to \R$ is a utility function to be specified in the following sections. The constraint is called the budget constraint. For a strictly concave utility $U$, the solution of Problem \eqref{prob:merton} is given by
\begin{equation}
	X_{\text{cla}} = I(\lambda_{\text{cla}} \rho),
\end{equation}
where the conjugate function $I:(0, \infty) \to \R$ is given by 
$
I(y) \triangleq \arg\sup_{x \in \R} \[ U(x) - xy \], ~y>0
$
(we will revisit the definition if $U$ is non strictly concave)
and $\lambda_{\text{cla}} > 0$ is a Lagrange multiplier solved from 
$ 
\E[\rho I(\lambda_{\text{cla}} \rho)] = \overline{x}.
$ 
The problem \eqref{prob:merton} can be seen as the terminal wealth optimization of the classic continuous-time Merton's problem in a complete market; see Appendix A of \cite{LL2024} for details. 
In the classic Merton's problem, the utility function is chosen as a smooth and strictly concave function, including power/log (CRRA, constant relative risk aversion) or exponential (CARA, constant absolute risk aversion) functions.


Now we introduce the concept of stochastic dominance (SD). For a random variable $X \in L^0$, the (upper) quantile function $Q_X: [0, 1] \to \R \cup \{\pm \infty\}$ is defined by
$$
Q_X(s) = \inf \{x \in \R \big| \p(X \leq x) > s \}, ~~ s \in [0, 1].
$$  
Denote by $\Q$ the set of all quantile functions:
$$
\Q \triangleq \left\{Q:[0,1] \to \R \cup \{\pm \infty\} \big| Q \text{ is increasing and right-continuous}\right\}.
$$
\begin{definition}[Stochastic dominance]\label{def:SD}
	\begin{enumerate}[(1)]
		\item Fix $X, Y \in L^0$. $X$ is greater than $Y$ in first-order stochastic dominance (FSD) if $Q_X(s) \geq Q_Y(s)$ for all $s \in (0, 1)$, which is denoted by $X \succeq_{(1)} Y$. 
		\item Fix $X, Y \in L^1$. $X$ is greater than $Y$ in second-order stochastic dominance (SSD) if $ \int_{0}^{t} Q_X(s) \d s \geq \int_{0}^{t} Q_Y(s) \d s$ for all $t \in (0, 1)$, which is denoted by $X \succeq_{(2)} Y$. 
	\end{enumerate}
\end{definition}


We specify a given benchmark wealth $X_0 \in L^0$. 
We proceed to study the problem with the first-order stochastic dominance (FSD) or second-order stochastic dominance (SSD) constraints:
\begin{equation}\label{prob:main1}
	\begin{aligned}
		\text{(FSD Problem) } & \max_{X \in L^0} 
		\E[U(X)] \text{ s.t. } \E[\rho X] \leq \overline{x} \text{ and } X \succeq_{(1)} X_0,
	\end{aligned}
\end{equation}
and
\begin{equation}\label{prob:main2}
	\begin{aligned}
		\text{(SSD Problem) } & \max_{X \in L^1} \E[U(X)] \text{ s.t. } \E[\rho X] \leq \overline{x} \text{ and } X \succeq_{(2)} X_0.
	\end{aligned}
\end{equation}


We denote the quantile function of $X_0$ by $Q_0 \in \Q$. Further, we define a minimal budget value:
$$
x_{Q_0} \triangleq \E[\rho X_0] = \int_{0}^{1} Q_0(s) Q_\rho(1-s) \d s.
$$
Throughout, we assume $\overline{x} \geq x_{Q_0}$. As a result, both problems \eqref{prob:main1}-\eqref{prob:main2} have at least one feasible solution $X_0$.

\cite{WX2021} propose and solve the FSD and SSD problems \eqref{prob:main1}-\eqref{prob:main2} with smooth and concave utilities; see \cite{WWX2024} for a mean-stochastic-dominance problem. In the following, we investigate the corresponding general non-concave utility optimization, particularly, S-shaped utility optimization.

\section{FSD Problem and Analytical Solution}\label{sec:FSD}

In this section, we apply the general S-shaped utility in Definition \ref{def:Sshaped} (see \cite{LL2020}) and proceed to study Problem \eqref{prob:main1}. 
Now we define a general S-shaped utility function $U$.
\begin{definition}[General S-shaped utility]
	\label{def:Sshaped}
	A general S-shaped utility function $U: \R \rightarrow \R$ with the reference point $B\in \R$ has the expression
	\begin{equation}\label{label:SUtility}
		U(x) = \left\{
		\begin{aligned}
			& U_1(x), && x \geq B,\\
			& U_2(x), && x < B,
		\end{aligned}
		\right.
	\end{equation}
	and satisfies the following properties:
	\begin{enumerate}[(i)]
		\item
		$U$ is increasing on $\R$, $U = U_1$ is differentiable and strictly concave on $(B, \infty)$, and $U = U_2$ is convex on $(-\infty, B)$.
		\item
		$U_1(B) = U_2(B)$ and $U_1'(B+) = U_2'(B-) = \infty$.
		\item
		Inada condition: $U_1'(\infty) = 0$.
	\end{enumerate}
\end{definition}


Note that in Definition \ref{def:Sshaped}, there is no requirement for a finite left endpoint of the domain of the utility function (known as the liquidation boundary in \cite{HK2018}). We solve the FSD problem \eqref{prob:main1} with a general S-shaped utility and hence illustrate that using the FSD constraint can replace the liquidation boundary for risk management. 
\begin{theorem}\label{thm:FSD}
	With a general S-shaped utility in Definition \ref{def:Sshaped}, the optimal solution of Problem \eqref{prob:main1} is given by
	\begin{equation}\label{eq:X*}\scriptsize
		X_{\text{FSD}}^* = \left\{
		\begin{aligned}
			& (U_1')^{-1}(\lambda \rho), \text{if } \left\{\rho < \frac{1}{\lambda} U_1' \left(Q_0(1 - F_\rho(\rho))\right) \text{ and } Q_0\left(1 - F_\rho(\rho)\right) \geq B  \right\} \text{ or } \left\{\rho \leq \frac{1}{\lambda} U_1' (C) \text{ and } Q_0\left(1 - F_\rho(\rho)\right) < B \right\};\\
			& Q_0(1 - F_\rho(\rho)), \text{if } \left\{ \rho \geq \frac{1}{\lambda} U_1' \left(Q_0(1 - F_\rho(\rho))\right) \text{ and } Q_0\left(1 - F_\rho(\rho)\right) \geq B \right\} \text{ or } \left\{ \rho > \frac{1}{\lambda} U_1' (C) \text{ and }  Q_0\left(1 - F_\rho(\rho)\right) < B \right\}.
		\end{aligned}
		\right. 
	\end{equation}
	where (i) the Lagrange multiplier $\lambda > 0$ is solved from the binding budget constraint $\E[\rho X_{\text{FSD}}^*] = \overline{x}$, and (ii) for any $Q_0\left(1 - F_\rho(\rho)\right) < B$, 
	the (state-dependent) tangent point $C \in (B, \infty)$ is solved from
	\begin{equation}\label{eq:tangent}
		\frac{U_1(C) - U_2(Q_0(1-F_\rho(\rho)))}{C - Q_0(1-F_\rho(\rho))} = U_1'(C).
	\end{equation}
\end{theorem}

\begin{proof}[Proof of Theorem \ref{thm:FSD}]
	
	As the objective in Problem \eqref{prob:main1} is an increasing function of $X$, the optimal wealth $X_{\text{FSD}}^*$ of Problem \eqref{prob:main1} is counter-monotonic to the pricing kernel $\rho$ (see, e.g., \cite{HZ2011,W2018}). Denote by $\xi$ the quantile transformation of $\rho$ such that $\xi$ has the uniform distribution on $[0, 1]$ and $Q_\rho(\xi) = \rho$. Hence, the optimal wealth $X_{\text{FSD}}^*$ under the first-order stochastic dominance constraint should satisfy
	$$
	X_{\text{FSD}}^* \geq Q_0(1-\xi).
	$$
	Hence, Problem \eqref{prob:main1} is translated to the following problem:
	\begin{equation}\label{prob:QuantForm}
		\max_{X \geq Q_0(1-\xi), \E\left[\rho X\right] = x} \E\left[ U(X) \right].
	\end{equation}
	Further, Problem \eqref{prob:QuantForm} is converted to a state-dependent pointwise optimization problem:
	\begin{equation}\label{prob:pointwise}
		\max_{X \geq Q_0(1-\xi)} \left\{ U(X) - \lambda \rho X  \right\},
	\end{equation}
	where $\lambda>0$ is a Lagrange multiplier to be determined such that $\E[\rho X] = \overline{x}$. To solve Problem \eqref{prob:pointwise}, we have the following two cases:
	\begin{enumerate}[(i)]
		\item for any $\omega \in \Omega$ satisfying $Q_0(1-\xi(\omega)) < B$, we solve the tangent point $C(\omega)$ from \eqref{eq:tangent} and have
		$$
		X_{\text{FSD}}^* = \left\{
		\begin{aligned}
			& (U_1')^{-1}(\lambda \rho), &&\text{if } (U_1')^{-1}(\lambda \rho) > C \text{ and } Q_0(1-\xi) < B;\\
			& Q_0(1-\xi), &&\text{if }  (U_1')^{-1}(\lambda \rho) \leq C \text{ and } Q_0(1-\xi) < B.
		\end{aligned}
		\right. 
		$$
		
		\item for any $\omega \in \Omega$ satisfying $Q_0(1-\xi(\omega)) \geq B$, we have
		$$
		X_{\text{FSD}}^* = \left\{
		\begin{aligned}
			& (U_1')^{-1}(\lambda \rho), &&\text{if } (U_1')^{-1}(\lambda \rho) > Q_0(1-\xi) \text{ and } Q_0(1-\xi) \geq B;\\
			& Q_0(1-\xi), &&\text{if }  (U_1')^{-1}(\lambda \rho) \leq  Q_0(1-\xi) \text{ and } Q_0(1-\xi) \geq B.
		\end{aligned}
		\right. 
		$$
	\end{enumerate} 
	Further, as $\xi = F_\rho(\rho)$, we derive the optimal solution $X_{\text{FSD}}^*$ given by \eqref{eq:X*}.
\end{proof}

In the literature, a liquidation boundary is needed for the optimization of the S-shaped utility, otherwise Problem \eqref{prob:main1} has no solution (mathematically, it is because one cannot establish a concave envelope for the S-shaped utility without a lower bound in the domain). In the FSD problem, we do not require the liquidation boundary for the S-shaped utility. The solution is twofold. In some good scenarios, it behaves like the classic solution. In some bad scenarios, it behaves like the benchmark quantile. From the solution, we can see that the FSD constraint acts as a good substitute for the liquidation boundary. 

\section{SSD Problem}\label{sec:SSD}
\subsection{SSD Problem under Non-concavity: Analytical Difficulty}

Let us restate the results of \cite{WX2021} on strictly concave utilities. 

\begin{lemma}[Theorem 5.10 of \cite{WX2021}]\label{thm:SSD}
	Let $\overline{x} > x_0$. 
    For a strictly concave utility $U$ with appropriate regularity conditions, 
	  the optimal solution to Problem \eqref{prob:main2} is $X_{\text{SSD}}^* = Q_{\text{SSD}}^*(1-F_\rho(\rho))$ with
	\begin{equation}\label{eq:Q^*}
		Q_{\text{SSD}}^*(s) = I\(\lambda \(Q_\rho(1-s) - y_{\text{SSD}}^*(1-s) \) \), ~ s \in (0,1),
	\end{equation}
	where $y_{\text{SSD}}^*: (0, 1) \to [0, \infty)$ is a function 
	given by the system
	\begin{equation}\label{eq:system}
		\left\{
		\begin{aligned}
			&y_{\text{SSD}}^* \text{ is right-continuous and } 0 \leq \frac{\d y_{\text{SSD}}^*(t)}{\d Q_\rho(t)} \leq 1, \text{  } Q_\rho(t) - y_{\text{SSD}}^*(t) > 0 \text{ for all } t \in (0, 1);\\
			&z_{\text{SSD}}^*(s) \triangleq -\int_{s}^{1} \[ I(\lambda (Q_\rho(t) - y_{\text{SSD}}^*(t))) - Q_0(1-t) \] \d t, ~s \in (0,1);\\
			&\frac{\d y_{\text{SSD}}^*(t)}{\d Q_\rho(t)} \left\{
			\begin{aligned}
				& \in [0, 1], && z_{\text{SSD}}^*(t) = 0;\\
				& = 0, && z_{\text{SSD}}^*(t) < 0,
			\end{aligned}
			\right. \quad dQ_\rho \text{-a.e.,}
		\end{aligned}
		\right.
	\end{equation}
	and the Lagrange multiplier $\lambda > 0$ is determined by the binding budget constraint equation
	\begin{eqnarray}\label{eq:budget}
		\int_{0}^{1} Q_{\text{SSD}}^*(s) Q_\rho(1-s) \d s = \overline{x}.
	\end{eqnarray}
\end{lemma}
	
In the proof of \cite{WX2021}, the procedure of solving Problem \eqref{prob:main2} starts with a view of quantile formulation (\cite{HZ2011,XZ2016,X2016}). 
Specifically, we define
$$
\begin{aligned}
	& \W \triangleq \{w: [0, 1] \to [0, \infty) | w(0)=0, w \text{ is increasing and concave} \},\\
	& \W_1 \triangleq \{w \in \W | w(1) = 1 \}.
\end{aligned}
$$
Define
\begin{equation}
	\begin{aligned}
		\Q_2(Q_0) &\triangleq \left\{ Q \in \Q \big| \int_{[0,1]} Q(s) \d w(s) \geq \int_{[0,1]} Q_0(s) \d w(s) \text{ for all } w \in \W \right\}\\
		&= \left\{ Q \in \Q \big| \int_{[0,1]} Q(s) \d w(s) \geq \int_{[0,1]} Q_0(s) \d w(s) \text{ for all } w \in \W_1 \right\}.
	\end{aligned}
\end{equation}
In view of the quantile formulation approach, we are going to reformulate the optimization over the random variable $X$ in terms of its quantile function $Q_X$. We can hence express the objective in SSD Problem \eqref{prob:main2} as
$
\E[U(X)] = \int_{0}^{1} U(Q_X(s)) \d s.
$ 
Based on the counter-monotonic dependence between the optimal solution $X$ and the pricing kernel $\rho$ (see \cite{HZ2011} and the proof of Theorem \ref{thm:FSD}), we can express the budget constraint as 
$
\E[\rho X] = \int_{0}^{1} Q_X (s) Q_\rho(1-s) \d s.
$ 
According to \cite{FS2016}, SSD Problem \eqref{prob:main2} reads as
an optimal quantile problem
\begin{equation}\label{prob:SSD_qf}
	\begin{aligned}
		& \max_{Q \in \Q_2(Q_0)} \int_{0}^{1} U(Q(s)) \d s ~~ \text{s.t. } \int_{0}^{1} Q(s) Q_\rho(1-s) \d s \leq \overline{x}.
	\end{aligned}
\end{equation}
The next step of solving Problem \eqref{prob:main2} is a conversion from 
Problem \eqref{prob:SSD_qf} to
\begin{equation}\label{prob:SSD}
	\begin{aligned}
		& \max_{Q \in \Q_2(Q_0)} \int_{0}^{1} U(Q(s)) \d s ~~
		 \text{ s.t. }\left\{ 
		\begin{aligned}
			& \int_{0}^{1} Q(s) Q_\rho(1-s) \d s \leq \bar{x},\\
			& \inf_{w \in \W_1} \(\int_{[0,1]} (Q(s) - Q_0(s)) \d w(s)\) \geq 0.
		\end{aligned}
		\right.
	\end{aligned}
\end{equation}
If the optimal solution of the original problem satisfies the second constraint in Problem \eqref{prob:SSD}, then the optimal solution of Problem \eqref{prob:SSD} is the same as that of the original problem and it is considered as a trivial case. 
In the non-trivial case, for any $\lambda > 0$, $Q \in \Q$ and $w \in \W$, we let 
$$
K(Q, w; \lambda) = \int_{0}^{1} U(Q(s)) \d s - \lambda \int_{0}^{1} Q(s) Q_\rho(1-s) \d s + \int_{[0,1]} (Q(s) - Q_0(s)) \d w(s).
$$ 
One needs to consider the following max-min problem 
for $K$: 
\begin{equation}
	\max_{Q \in \Q} \min_{w \in \W} K(Q, w; \lambda).
\end{equation}
Unfortunately, the solution procedure of the SSD Problem under the non-concave utility is stuck at this step, because the desired Sion's max-min theorem requires that $K$ is concave in $Q$, which does not hold generally. In an alternative clue of concavifying $U$ (see \cite{LL2020}), one cannot guarantee that the concavification principle is valid (i.e., the optimal wealth variable under the concave envelope is almost surely not located in the region where the original utility and its concave envelope do not coincide); here, the concave envelope is defined as the smallest concave function dominating $U$. Even if a similar form of Lemma \ref{thm:SSD} is established for general non-concave utilities, one needs to solve the optimal pair $(y^*, z^*)$ from the system \eqref{eq:system}, which is an infinite-dimensional optimization problem over the functional space. \cite{WX2021} propose some explicit optimal solutions based on specific and technical assumptions on $U$, $Q_\rho$ and $Q_0$. Beyond these, there is no general analytical expression for optimal solutions. 

Nevertheless, Theorem 5.10 of \cite{WX2021} provides an important idea that the optimal solution may be characterized by the optimal pair $(y^*, z^*)$ from the system \eqref{eq:system}, and hence we can obtain some heuristics to construct suboptimal solutions based on numerical algorithms and even neural networks.  Specifically, we define a concept of the \textbf{poor performance region}:
$$
C \triangleq \left\{ t \in (0, 1) | I(\lambda Q_\rho(t)) - Q_0(1-t) < 0 \right\}.
$$
Here for the non-concave utility $U$, we define
$I(y) \triangleq \inf \{ \arg\sup_{x \in \R} \[ U(x) - xy \] \}, ~y>0$ whenever applicable. 
\begin{remark} 
``Whenever applicable" means that the function $I(y)$ is finite for any $y\in (0, \infty)$. That is, the utility $U$ has a finite concave envelope function. In this case, $I(\cdot)$ is right-continuous and decreasing on $(0, \infty)$. For example, if $U$ is an S-shaped utility with the domain $[L, \infty)$ where $L\in \R$, we have $I(\cdot)$ is well-defined, while if the domain is $\R$, $I(\cdot) \equiv -\infty$.
\end{remark}

In the set $C$, we compare the unconstrained classic solution $I(\lambda Q_\rho(t))$ with the SSD benchmark $Q_0(1-t)$. If $C = \emptyset$, then it means that the unconstrained solution automatically satisfies the SSD constraint, and $Q_{\text{SSD}}^*$ should be the same as the unconstrained solution. Next, we discuss the non-trivial case (i.e., $C \neq \emptyset$). For scenarios $t$ on this region, the unconstrained solution $I(\lambda Q_\rho(t))$ is smaller (i.e., performing worse) than the SSD benchmark $Q_0(1-t)$. 
From Definition \ref{def:SD}, the SSD constraint $Q_{\text{SSD}}^* \succeq_{(2)} Q_0$ reads as
$$
\int_{0}^{s} Q_{\text{SSD}}^*(t) \d t \geq \int_{0}^{s} Q_0(t) \d t ~~ \text{ for any } s \in (0,1),
$$
which is reflected in the system \eqref{eq:system}:
$$
z_{\text{SSD}}^*(s) = -\int_{s}^{1} \( Q_{\text{SSD}}^*(1-t) - Q_0(1-t) \) \d t \leq 0 ~~ \text{ for any } s \in (0,1),
$$
which translates to 
\begin{equation}\label{eq:SSD_zstar}
	z_{\text{SSD}}^*(s) = -\int_{s}^{1} \[ I(\lambda (Q_\rho(t) - y_{\text{SSD}}^*(t))) - Q_0(1-t) \] \d t \leq 0 ~~ \text{ for any } s \in (0,1).
\end{equation}
The theorem inspires that some correction function $y_{\text{SSD}}^*$ should be added to satisfy the constraint \eqref{eq:SSD_zstar}. In the next subsection, we provide a numerical algorithm and design the correction function to obtain a sub-optimal solution.



\subsection{SSD Problem: Numerical Algorithm}\label{sec:algorithm}
\begin{algorithm}[htbp]\normalsize
	\caption{Poor-Performance-Region Algorithm for SSD Problem \eqref{prob:main2} with general utilities}\label{alg}
	\begin{algorithmic}[1]
		\State We solve the classic problem without the SSD constraint and obtain the optimal quantile $Q_{\text{cla}} (\cdot) \triangleq I(\lambda_{\text{cla}} Q_\rho(1 - \cdot))$. The Lagrange multiplier is denoted by $\lambda_{\text{cla}} \in (0, \infty)$, which is solved from the following equation
		\begin{equation*}
		\begin{aligned}
			\overline{x} &= \int_{0}^{1} Q_\rho(s) I(\lambda_{\text{cla}} Q_\rho(s)) \d s.
		\end{aligned}
		\end{equation*}
		If 
		$$
		-\int_{t}^{1} \left[ I(\lambda_{\text{cla}} Q_\rho(s)) - Q_0(1-s) \right] \d s \leq 0
		$$
		holds for any $t \in [0,1]$, then the optimal solution is $Q_{\text{cla}} (\cdot)$. Otherwise, we start the procedure below.
		
		\State The Lagrange multiplier $\lambda$ is initially set as the above (to be determined at last). Compute the set 
		$$
		C = \left\{ t \in (0, 1) | I(\lambda Q_\rho(t)) - Q_0(1-t) < 0 \right\}.
		$$
		
		\State If $C = \emptyset$, then the optimal solution is $Q_{\text{cla}}$. Otherwise, specify an appropriate $n \in \N$ and write $C = \cup_{i = 1}^n (a_i, b_i)$ with $a_i < b_i$. Further we set $a_{n+1} = 1$.
		
		%
		%
		%
        \State For $i = n, (n-1), \cdots, 1$ (Steps 4-6), we compute
        \begin{equation}\label{eq:alg_y0}
			y_0(s) = \inf \{y\geq0| I\left(\lambda (Q_\rho(s) - y)\right) - Q_0(1-s) \geq 0\}, ~~ \text{$s \in (a_i, b_i)$}.
		\end{equation}
        Define
        \begin{equation}\label{eq:gi}
			g_i(s) = - \int_{s}^{a_{i+1}} \[ I\(\lambda \(Q_\rho(t) - y_0(s) \) \) - Q_0(1-t) \] \d t, ~~ s \in (a_i, b_i).
		\end{equation}

        \State We compute
        \begin{equation}\label{eq:ti}
        t_i = \sup\Big\{\, t \in [a_i, b_i] \;\big|\; g_i(t) +  z_{\text{sub}}(a_{i+1}) \id_{\{i\neq n\}} > 0 \Big\},
        \end{equation}
        where $z_{\text{sub}}(\cdot)$ will be determined in Step 6.
        If $\{ t \in [a_i, b_i] | g_i(t) +  z_{\text{sub}}(a_{i+1}) \id_{\{i\neq n\}}> 0 \} = \emptyset$, set $t_i = a_{i}$.
		
		\State Set 
        \begin{equation}\label{eq:y_sub}
        y_{\text{sub}}(\cdot) \equiv y_0(t_i) \text{ on } (t_i, a_{i+1})~ ~\text{ and } ~~ y_{\text{sub}}(\cdot) = y_0(\cdot) ~\text{ on } (a_i, t_i).                 
        \end{equation}
        Define $z_{\text{sub}}(\cdot)$ as follows:
        \begin{equation}\label{eq:z_sub}
        z_{\text{sub}}(s) \triangleq -\int_{s}^{1} \[ I(\lambda (Q_\rho(t) - y_{\text{sub}}(t))) - Q_0(1-t) \] \d t, ~s \in [a_i, a_{i+1}].
        \end{equation}

		\State 
		Set $t_0 = 0$ and $y_0(t_0) = 0$. After the iteration, we have Eq. \eqref{eq:y_sub} for $i = n, \cdots, 1$ 
		and $y_{\text{sub}}(\cdot) \equiv y_0(t_0)$ on $(t_0, a_1]$. We then verify whether $y_{\text{sub}}(\cdot)$ satisfies the monotonicity condition (non-decreasing over $(0,1)$). If the condition holds, proceed to Step 14; otherwise, apply the correction procedure and proceed to Step 8.

        \State For $i = (n-1), \dots, 2$ (Steps 8-12), check whether $y_0(\cdot)$ is increasing (non-decreasing) over $(a_{i-1}, a_{i+1})$. If yes, skip and proceed to next iteration; if not, proceed to Step 9.

        \State Compute $y_0(\cdot)$ over $(a_{i-1}, b_{i-1})$ and $(a_i, b_i)$ by Eq. \eqref{eq:alg_y0}.
    
    
        \State Redefine
        $$g_i(s) = - \int_{s}^{\bar{s}} \[ I\(\lambda \(Q_\rho(t) - y_0(s) \) \) - Q_0(1-t) \] \d t, ~~ s \in (a_{i-1}, b_{i-1}),$$
        where $\bar{s} = \inf \{u \in [a_i, b_i) | y_0(u) - y_0(s) \geq 0\}.$
    
        \State
        Compute
        \begin{equation}\label{eq:t_left}
            t_{\text{left}} = \sup \{ t \in [a_1, b_1] | g_1(t) > 0 \}, \;\;t_{\text{right}} = \bar{s}.
        \end{equation}


	\end{algorithmic}
\end{algorithm}

\begin{algorithm}[htbp]\normalsize
    \begin{algorithmic}[1]
        \setcounter{ALG@line}{11}
        \State Replace the initial $y_{\text{sub}}(\cdot)$ over $(t_{\text{left}}, t_{\text{right}})$ and set $$y_{\text{sub}}(\cdot) \equiv y_0(t_{\text{left}}) \text{ on } (t_{\text{left}}, t_{\text{right}}).$$
    
        \State After the iteration, check whether $y_{\text{sub}}(\cdot)$ satisfy the monotonicity condition. If yes, proceed to next step; otherwise, the algorithm may fail.
        
        \State Hence, we design the sub-optimal solution by
		\begin{equation}\label{eq:Q_algo}
		Q_{\text{sub}}(s) = \left\{
		\begin{aligned}
			&I\(\lambda \( Q_\rho(1-s) - y_0(t_n)\)\), && s \in (1-a_{n+1}, 1-t_n];\\
			&I\(\lambda \( Q_\rho(1-s) - y_0(1-s)\)\) \equiv Q_0(s), && s \in (1-t_n, 1-a_n];\\
			&\cdots\\
			&I\(\lambda \( Q_\rho(1-s) - y_0(t_0)\)\), && s \in (1-a_1, 1-t_0).\\
		\end{aligned}
		\right.
		\end{equation}

        \State Set
		$$
		z^*_i(t) = -\int_{t}^{1} \[ I\(\lambda \(Q_\rho(s) - y_0(t_i) \) \) - Q_0(1-s) \] \d s, ~~ t \in (\max\{a_i, t_i\}, a_{i+1}).
		$$ 
		If for any $i = 1, \cdots, n$, $z^*_i(\cdot) \leq 0$ always holds on the interval $(\max\{a_i, t_i\}, a_{i+1})$, 
        this 
        $Q_{\text{sub}}$ is feasible and sub-optimal. 

		\State Using a bisection method, we determine $\lambda_{\text{sub}} \in (0, \infty)$ from the following equation
		\begin{equation}\label{eq:budget_binding}
		\begin{aligned}
			\overline{x} &= \int_{0}^{1} Q_\rho(s) Q_{\text{sub}}(1-s) \d s.
		\end{aligned}
		\end{equation}
    \end{algorithmic}
\end{algorithm}

Inspired by the structure of the optimal solution \eqref{eq:Q^*}, we proceed to present a numerical algorithm to propose a suboptimal solution $Q_{\text{sub}}$, where we design a correction function $y_{\text{sub}}$ based on the value of $z_{\text{sub}}$ in different sections of the poor performance region $C$.

Detecting the structure of the poor performance region is the key task. We first define a function
\begin{equation}\label{eq:y0}
	y_0(s) \triangleq \inf \{y \geq 0| I\left(\lambda (Q_\rho(s) - y)\right) - Q_0(1-s) \geq 0\}, ~~ s \in (0, 1).
\end{equation}
Hence, we alternatively write the poor performance region as 
\begin{equation}\label{eq:relation_C_y0}
    C = \{t \in (0, 1) | y_0(t) > 0 \}. 
\end{equation}
For some very poorly performing scenarios $t \in C$, the function $y_0$ is adopted such that $Q_{\text{sub}}(1-t) = I(\lambda (Q_\rho(t) - y_0(t))) = Q_0(1-t)$.

The general idea in the construction of the sub-optimal solution to SSD Problem \eqref{prob:main2} is: when $z_{\text{sub}} \leq 0$, we use the classic solution to achieve optimality (now $y_{\text{sub}} = 0$ or constant, and $Q_{\text{SSD}}$ is the form of $I$); when $z_{\text{sub}} > 0$, we set $y_{\text{sub}} = y_0$ such that $Q_{\text{sub}} = Q_0$ to satisfy the SSD constraint. Throughout the algorithm design, we need to guarantee that $y_{\text{sub}}$ is increasing and non-negative. 
As the optimal structure in Theorem 5.10 of \cite{WX2021} comes from a strictly concave utility setting, the sub-optimal solution under our setting of a non-concave utility $U$ is actually the ``solution" of the concave envelope of $U$.




We therefore propose Algorithm \ref{alg} below, named the \textit{Poor-Performance-Region Algorithm} (PPRA). Based on Theorem \ref{thm:SSD}, we design a closed-form sub-optimal solution to Problem \eqref{prob:SSD_qf}:
\begin{equation}
	Q_{\text{sub}}(s) = \sum_{i = 0}^n I\left(\lambda_{\text{sub}} \left(Q_\rho(1-s) - y_0(t_i)\right)\right) \id_{\{s \in (1-a_{i+1}, 1-t_i)\}} + \sum_{i = 1}^n Q_0(s) \id_{\{s \in [1-t_i, 1-a_i]\}}, \;\; s \in (0, 1).
\end{equation}
This quantile function is also given in Eq. \eqref{eq:Q_algo}. The key idea is to use an increasing step function $y^*_{\text{SSD}}(\cdot)$ for approximation. The financial insight is that the decision maker should follow the SSD benchmark quantile $Q_0$ on some poor performance scenario and conduct the unconstrained optimal strategy $I\left(\lambda_{\text{sub}} \left(Q_\rho(1-s) - y_0(t_i)\right)\right)$ otherwise.

Here are some explanations of the algorithm design.
\begin{enumerate}[a.]

    \item In this algorithm, the partition points ${\{a_i\}}^n_{i=1}$ and ${\{t_i\}}^n_{i=1}$ of the poor performance region play an essential role in determining the structure of the optimal solution.
    
	\item In Step 2: In the whole procedure, we are solving out the structure of the optimal solution and the Lagrange multiplier. For the latter, note that the initial Lagrange multiplier may not satisfy the budget constraint. But it is a good initial point to start the algorithm. It will be determined in the final step.
	
	\item In Step 3: Because $I(\lambda Q_\rho(\cdot))$ and $Q_0(1-\cdot)$ are both nonincreasing, the set $C$ can be written as the union of disjoint intervals $\cup_{i=1}^n (a_i, b_i)$ or $\cup_{i=1}^\infty (a_i, b_i)$. In the latter case, to construct a numerically tractable solution, we use the union of the first $n$ disjoint intervals, where $n$ can be specified based on one's computational capability. 
	
	\item In Steps 4-5: For any $t \in [b_i, a_{i+1})$, we have
	$$
	I\(\lambda^* \(Q_\rho(s) - y_0(t) \) \) - Q_0(1-s) \geq I\(\lambda^* Q_\rho(s) \) - Q_0(1-s) \geq 0, ~~ s \in (t, a_{i+1}].
	$$
	We then compute
	$$
	g_i(t) = - \int_{t}^{a_{i+1}} \[ I\(\lambda^* \(Q_\rho(s) - y_0(t) \) \) - Q_0(1-s) \] \d s \leq 0. 
	$$
	Hence, any $t \in [b_i, a_{i+1})$ satisfies the SSD constraint. We desire to search the first point $t_i$ on $(a_i, b_i)$ which does not satisfy the constraint.

    \item In Steps 8-12, we need to check whether the constructed $y_{\text{sub}}(\cdot)$ is non-decreasing on (0,1) to make the solution valid. It holds in many cases. In some invalid cases, we could still ensure the monotonicity by redesigning the construction of $t_i$ defined in Eq. \eqref{eq:ti}, which is named $t_{\text{left}}$ in Eq. \eqref{eq:t_left}. The specific case will be shown in Section \ref{sec:Numerical results}.
    
	\item In Step 14, in many cases, one can design that $Q_{\text{sub}}(\cdot) = I(\lambda (Q_\rho(1-\cdot) - y_0(\cdot))) = Q_0(\cdot)$ on $(1-t_n, 1-a_n]$.
    
	\item In Step 15, we need to check whether the condition holds numerically. It holds in many cases.

\end{enumerate}

In the later sections, based on Algorithm \ref{alg}, we are able to provide the numerical illustration for the SSD Problem \eqref{prob:main2} with various concrete settings.  

\section{Numerical Results}
\label{sec:Numerical results}

Our study is motivated by the Black-Scholes model in a complete market. A classic Black-Scholes model consists of one riskless bond ($\frac{\d B_t}{B_t} = r \d t, ~ t \in [0, T]$, where the risk-free rate is $r = 0.05$) and one stock ($\frac{\d S_t}{S_t} = \mu_{\text{S}} \d t + \sigma_{\text{S}} \d W_t, ~ t \in [0, T]$, which represents a geometric Brownian motion with the expected return rate $\mu_{\text{S}} = 0.086$, the volatility parameter $\sigma_{\text{S}} = 0.3$ and a standard Brownian motion
$\{W_t\}_{0 \leq t \leq T}$). 
The wealth process is given by
$\d X_t = \(r X_t + (\mu_S-r) \pi_t \) \d t + \sigma_S\pi_t \d W_t, ~ t \in [0, T]  
$ 
and 
$X_0 = \overline{x}$, 
where $\{\pi_t\}_{0 \leq t \leq T}$ is the control process representing the investment amount in the stock and $T \in (0, \infty)$ is the evaluation time of investment. We define the pricing kernel process $\{\rho_t\}_{0 \leq t \leq T}$  by
\begin{equation}\label{eq:rho}
	\begin{aligned}
		& \frac{\d \rho_t}{\rho_t} = -r \d t - \theta \d W_t, ~~ t \in [0, T],
	\end{aligned}
\end{equation}
where we denote the market price of risk by $\theta \triangleq (\mu_S - r)/\sigma_S$. 

As the market is complete and one can use the martingale method to duplicate the optimal portfolio process, it is sufficient to solve the optimal terminal wealth variable via the static problem \eqref{prob:merton}  (see, e.g., Appendix A of \cite{LL2024}). Hence, our focus is solving the optimal wealth variable in Problem \eqref{prob:merton}. 
Adapting to Problem \eqref{prob:merton}, we denote the terminal variables of pricing kernel and wealth by $\rho := \rho_T$ and $X := X_T$, with a slight abuse of notation. Solving Eq. \eqref{eq:rho}, $\rho$ follows the log-normal distribution (i.e., $\log(\rho) \sim \text{N}(- (r + \theta^2/2)T, (\theta \sqrt{T})^2)$) and has a quantile function
\begin{equation}\label{eq:Qrho_BS}
Q_\rho(t) =  \exp \left\{\theta \sqrt{T} \Phi^{-1}(t) - (r + \theta^2/2)T) \right\} \triangleq \exp \left\{\sigma \Phi^{-1}(t) + \mu \right\}, ~ t \in [0, 1],
\end{equation}
where we denote by $\Phi^{-1}$ the standard normal quantile function and define $\sigma \triangleq \theta \sqrt{T}$ and $\mu \triangleq -(r + \theta^2/2)T)$. 
In Sections \ref{sec:Numerical results}-\ref{sec:NN}, we will mainly consider this $Q_\rho$ in Eq. \eqref{eq:Qrho_BS} and specify different benchmark quantile functions $Q_0$ and various utility functions.

We specify the parameters:
 the risk-free rate is $r = 0.05$;  
the expected return rate $\mu_{\text{S}} = 0.086$;  
the volatility parameter $\sigma_{\text{S}} = 0.3$;  
the evaluation time of investment is $T = 20$ (years).   
It follows that the market price of risk is $\theta = (\mu - r)/\sigma = 0.12$. We compute that $\sigma = 0.5367$ and $\mu = -1.1440$. 
We numerically illustrate our result by using the Black-Scholes model above. 

\subsection{FSD Problem: S-shaped Utility (Theorem \ref{thm:FSD})}

We begin by specifying an S-shaped utility $U:[L, \infty) \rightarrow \R$, following the general S-shaped utility formulation in Definition \ref{def:Sshaped}:

\begin{equation}\label{eq:Sshaped}
	\begin{aligned}
		&U(x) = \left\{
		\begin{aligned}
			& \frac{x^p}{p}, && x \geq 0,\\
			& -k (-x)^q, && L \leq x < 0,
		\end{aligned}
		\right.
	\end{aligned}
\end{equation}
where the parameters are set to $p = 0.6, q=0.5$, and $ k=2$.

 We then consider two portfolio selection problems. The first is Problem \eqref{prob:main1} with the general setting of S-shaped utility in Definition \ref{def:Sshaped} (in particular, Eq. \eqref{eq:Sshaped} with $L=-\infty$) and the benchmark quantile 
$
Q_0(t) = 10 t^2 - 1, \; t \in [0, 1].
$
The second is the Merton problem \eqref{prob:merton} using the S-shaped utility setting in Eq. \eqref{eq:Sshaped}:
\begin{equation}\label{prob:liq}
	\begin{aligned}
		& \max_{X \in L^0} \E[U(X)], \text{ s.t. } \E[\rho X] \leq \overline{x}, \; X \geq L \text{ a.s., }
	\end{aligned}
\end{equation}
where the liquidation boundary is given by $L = -5$ and $\bar{x}=5$. Here we add the liquidation boundary in order to make the second problem well defined and compare this classic solution with the first problem.

In this example, even though the first Problem \eqref{prob:main1} does not admit a liquidation constraint ($L=-\infty$), the first-order SD constraint acts a similar role: if $Q_0(1- F_\rho(\rho)) = Q_0(0)$, then the optimal solution $X^*$ is located at the boundary $Q_0(1- F_\rho(\rho))$, otherwise $X^*$ is the same as the classic solution $(U_1')^{-1}(\lambda \rho)$.
\begin{figure}[H]
	\centering
	\includegraphics[width=0.4\textwidth]{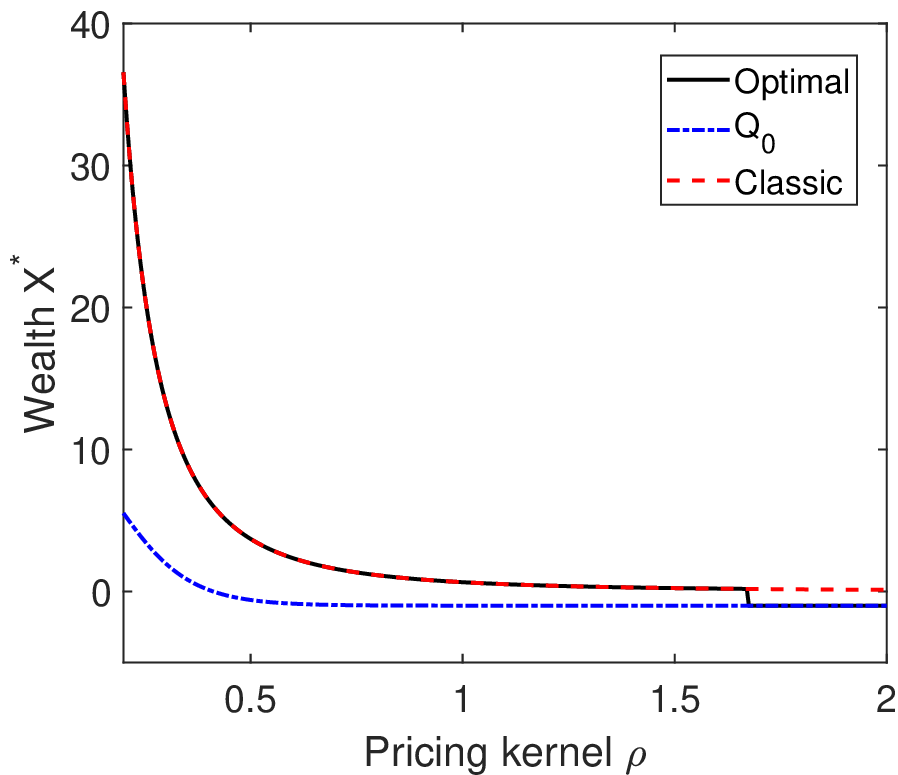}
	\caption{FSD problem: $\eqref{eq:Sshaped}$ vs $\eqref{prob:liq}$.}
	\label{fig:FSD_square}
\end{figure}

\subsection{SSD Problem: Power Utility}

We first assume the decision maker has a CRRA utility
\begin{equation}\label{eq:powerU}
	U(x) = \frac{1}{p} x^{p}, ~~ x > 0,
\end{equation}
where $p = 0.6$. We suppose that the benchmark quantile $Q_0$ also follows the log-normal distribution:
\begin{equation}\label{eq:ssd_Q0}
    Q_0(t) =  \exp \left\{\sigma_0 \Phi^{-1}(t) + \mu_0 \right\}, \; t \in [0, 1].
\end{equation}

The settings are given in Table \ref{table:power}.

\begin{table}[h]
	\centering
	\begin{tabular}{c c c c c c c c c c c c c c c c c }
		\hline
		$r$ & $\mu_{\text{S}} $ & $\sigma_\text{S}$ & $\theta$ & $T$ & $\mu$ & $\sigma$ & $\overline{x}$ & $p$   \\
		\hline
		0.05 & 0.086 & 0.3 & 0.12 & 20 & -1.1440 & 0.5367 & 10 & 0.6  \\
		\hline
	\end{tabular}
	\caption{Parameter settings in the numerical illustration.}
	\label{table:power}
\end{table}
For the concave utility specified in Eq. \eqref{eq:powerU}, the applicability of Proposition 6.8 in \cite{WX2021} to power utility with a log-normal pricing kernel provides a theoretical benchmark. Our algorithm’s result coincides with the characterization in \cite{WX2021}. Consequently, the algorithm attains the optimal solution in the following cases in Figure \ref{fig:SSD_power}.
For comparison, we also compute the classical optimal solution obtained without the SSD constraint and plot the figure. 
%

\begin{figure}[htbp]
	\begin{gather*}
		\tiny
		\begin{matrix}
			\includegraphics[width=0.48\textwidth]{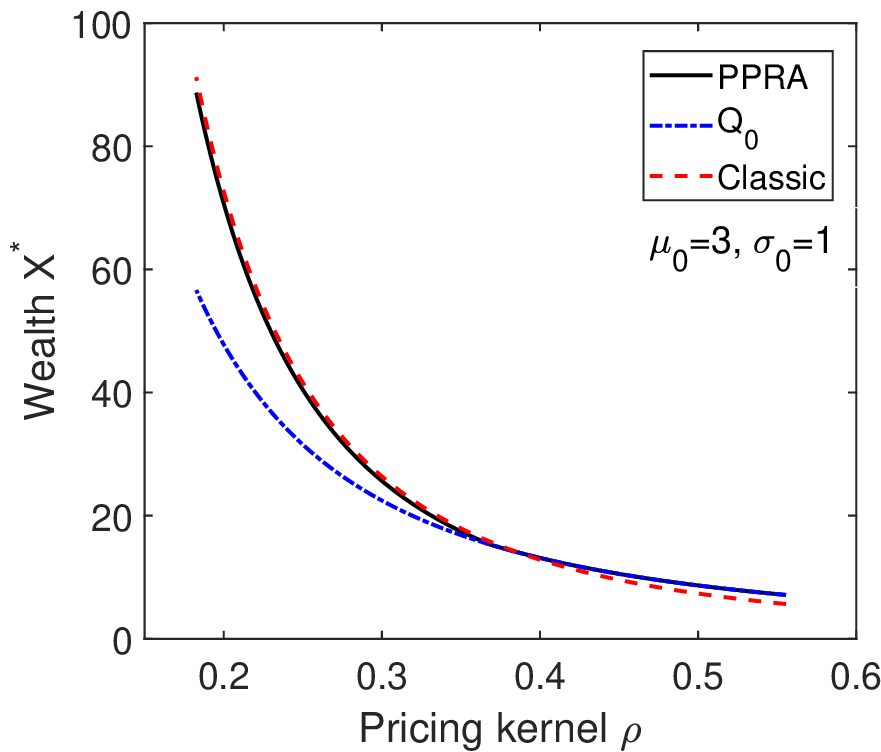}
			& \includegraphics[width=0.48\textwidth]{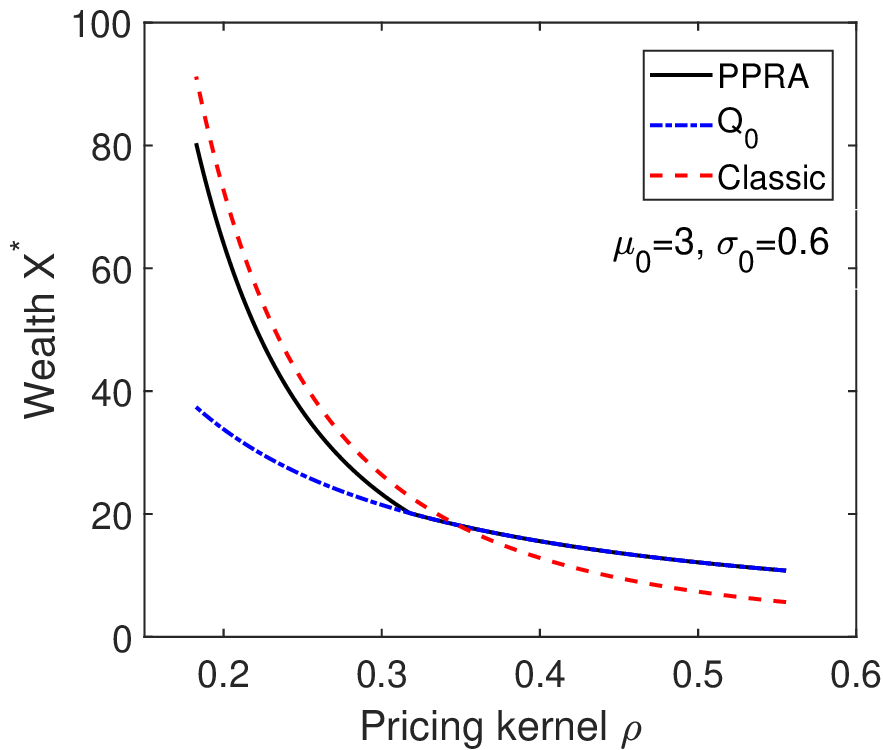}
		\end{matrix}
	\end{gather*}
	\begin{gather*}
		\tiny
		\begin{matrix}
			\includegraphics[width=0.48\textwidth]{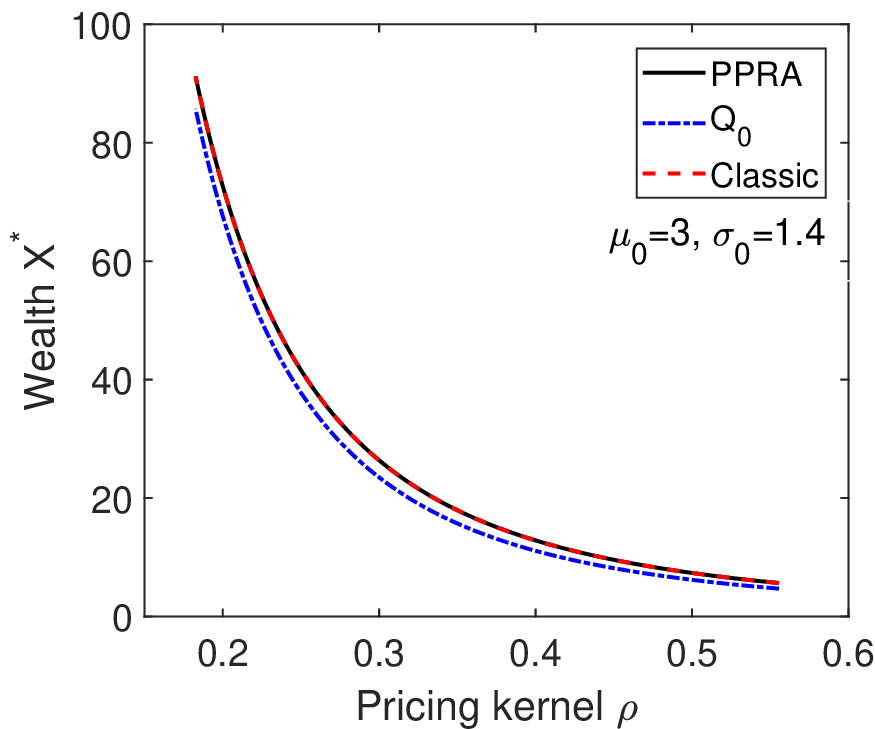}
			& \includegraphics[width=0.48\textwidth]{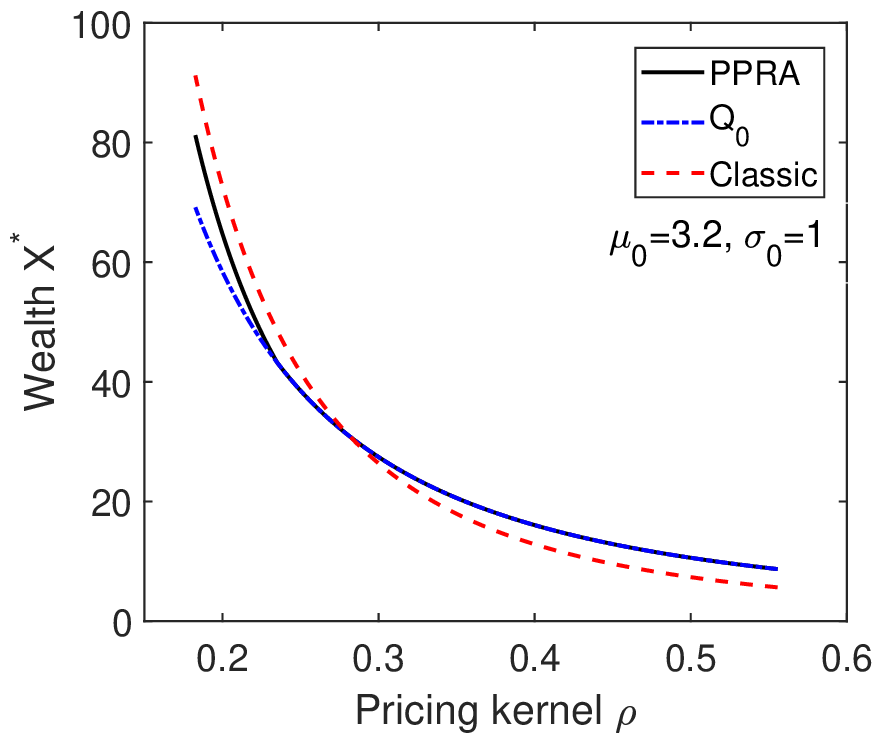}
		\end{matrix}
	\end{gather*}
	\begin{gather*}
		\tiny
		\begin{matrix}
			\includegraphics[width=0.48\textwidth]{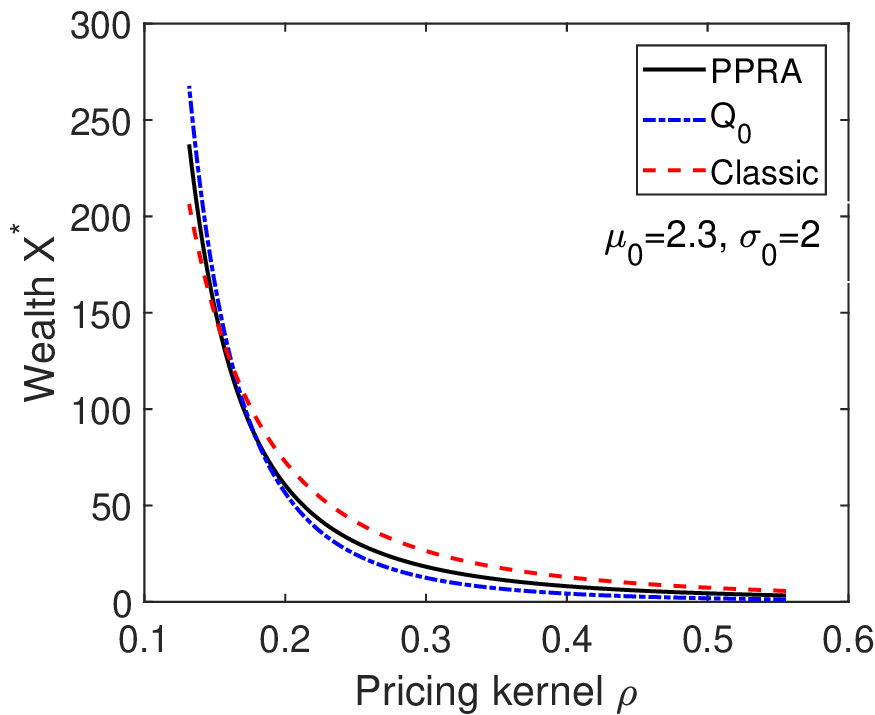}
			& \includegraphics[width=0.48\textwidth]{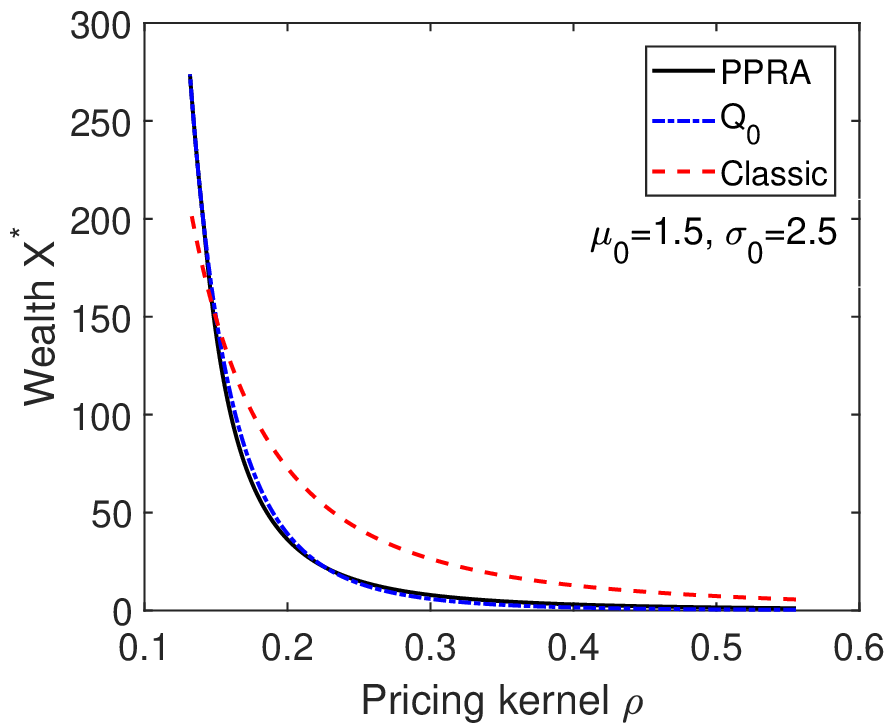}
		\end{matrix}
	\end{gather*}
	\caption{SSD Problem: Setting $\eqref{eq:powerU}$ and $\eqref{eq:ssd_Q0}$. }
	\label{fig:SSD_power}
\end{figure}

\begin{table}[h]
	\centering
	\begin{tabular}{c|c|c|c|c|c|c}
		\hline
		&parametrization & budget of $Q_0$ & poor performance region $C$ & $\lambda$ & $\lambda_\text{cla}$ & partition parameter \\
		\hline
		(a) & $(\mu_0, \sigma_0) = (3, 1)$ & 7.1231 & $(0.6092,1)$  & $0.9104$  & 0.9003 & $t_1 = 1$ \\
		\hline
		(b) & $(\mu_0, \sigma_0) = (3, 0.6)$ & 6.4109 & $(0.4978, 1)$  & 0.9471  & 0.9003 & $t_1 = 1$ \\
		\hline
		(c) & $(\mu_0, \sigma_0) = (3, 1.4)$ & 9.2876 & (0, 0.0179) & 0.9003 & 0.9003 & $t_1 = 0$\\
		\hline
		(d) & $(\mu_0, \sigma_0) = (3.2, 1)$ & 8.7002 & (0.2858, 1) & 0.9430 & 0.9003 & $t_1 = 1$\\
		\hline
		(e) & $(\mu_0, \sigma_0) = (2.3, 2)$ & 9.2691 & (0, 0.4309) & 1.1951 & 0.9003 & $t_1 = 0.0057$\\
		\hline
		(f) & $(\mu_0, \sigma_0) = (1.5, 2.5)$ & 9.8096 & (0, 0.6248) & 1.9965 & 0.9003 & $t_1 = 0.0654$\\
		\hline
	\end{tabular}
	\caption{Numerical results in Figure \ref{fig:SSD_power}.}
	\label{nume:power}
\end{table}

The basic logic is that: 
\begin{enumerate}[(i)]
	\item If the value of the pricing kernel is small, then the optimal wealth value is larger. Hence, the pricing kernel value is a signal of the market state: A small value shows a good market scenario. 
	
	\item If the poor performance region $C$ is smaller, then the optimal wealth is better (compared to the SSD constraint). This is because the SSD constraint is easier to achieve and the optimal wealth is closer to the classic unconstrained solution $X_\text{cla}$. 
\end{enumerate}
We apply the Poor-Performance-Region Algorithm (PPRA), and the explanations and financial insights from Figure \ref{fig:SSD_power} and Table \ref{nume:power} are as follows:
\begin{enumerate}

	\item Figure \ref{fig:SSD_power} (a) vs (b) vs (d): When the market performs poorly, the optimal wealth must coincide with the benchmark. This is because, in such adverse scenario, the benchmark provides a large value, and the SSD constraint serves to guarantee a minimum safety level and reduce risk.

    \item Figure \ref{fig:SSD_power} (c): The poor performance region is very small, which implies the classical solution inherently satisfies the SSD constraint. As a result, we observe that $\lambda = \lambda_{\text{cla}}$ and the optimal solution essentially coincides with the classical solution. This indicates that the benchmark plays only a limited role in shaping the optimal wealth.

    \item Figure \ref{fig:SSD_power} (e) vs (f): The budget of $Q_0$ is close to $\bar{x}$, to ensure the SSD constraint, the optimal wealth would behave similarly to the benchmark. However, due to a different $\lambda$ and a correction function $y_{\text{sub}}(\cdot)$ in the SSD problem, the solution between the SSD problem and the classical problem differs.

	
	
	
	
\end{enumerate}

\subsection{SSD Problem: S-shaped Utility}\label{subsection:Sshaped}
We use Algorithm \ref{alg} to study the sub-optimal solution of the SSD problem with S-shaped utility. First, we adopt the S-shaped utility setting in Eq. \eqref{eq:Sshaped}.
We then suppose that the benchmark quantile $Q_0$ also follows a log-normal distribution:
$$
Q_0(t) =  \exp \left\{\sigma_0 \Phi^{-1}(t) + \mu_0 \right\}, \; t \in [0, 1].
$$
We also compute the classic optimal solution without the SSD constraint. The settings are provided in Table \ref{table:SU}.
\begin{table}[h]
	\centering
	\begin{tabular}{c c c c c c c c c c c c c c c c c }
		\hline
		$r$ & $\mu_{\text{S}} $ & $\sigma_\text{S}$ & $\theta$ & $T$ & $\mu$ & $\sigma$ & $\overline{x}$ & $p$ & k  \\
		\hline
		0.05 & 0.086 & 0.3 & 0.12 & 20 & -1.1440 & 0.5367 & 10 & 0.6 & 2 \\
		\hline
	\end{tabular}
	\caption{Parameter settings in Figure \ref{fig:SSD_SU}.}
	\label{table:SU}
\end{table}

\begin{figure}[htbp]
	\begin{gather*}
		\tiny
		\begin{matrix}
			\includegraphics[width=0.42\textwidth, height=0.29\textheight]{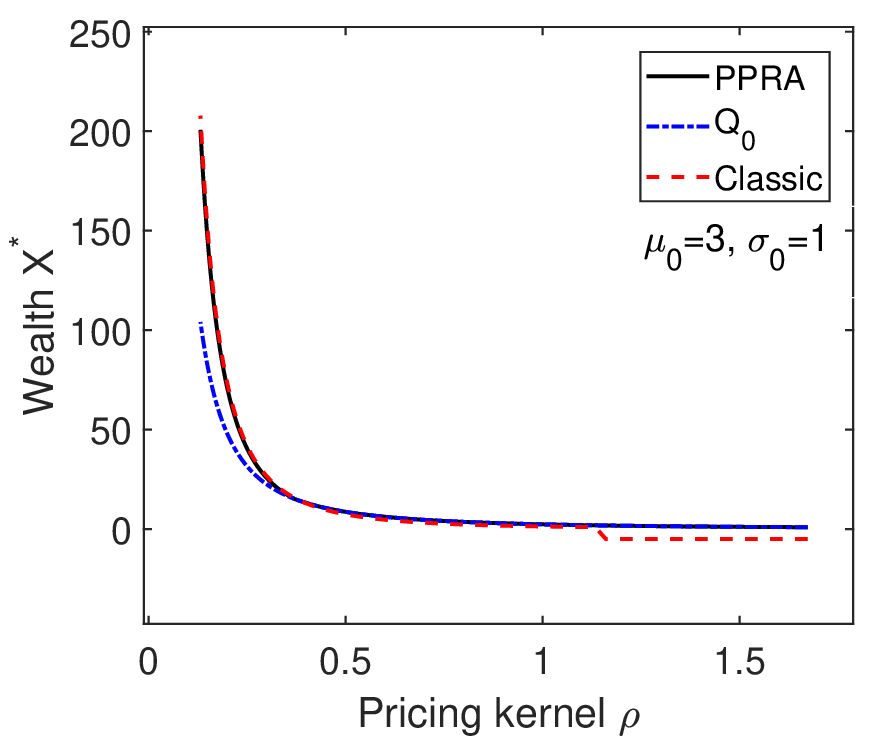}
			& \includegraphics[width=0.42\textwidth, height=0.29\textheight]{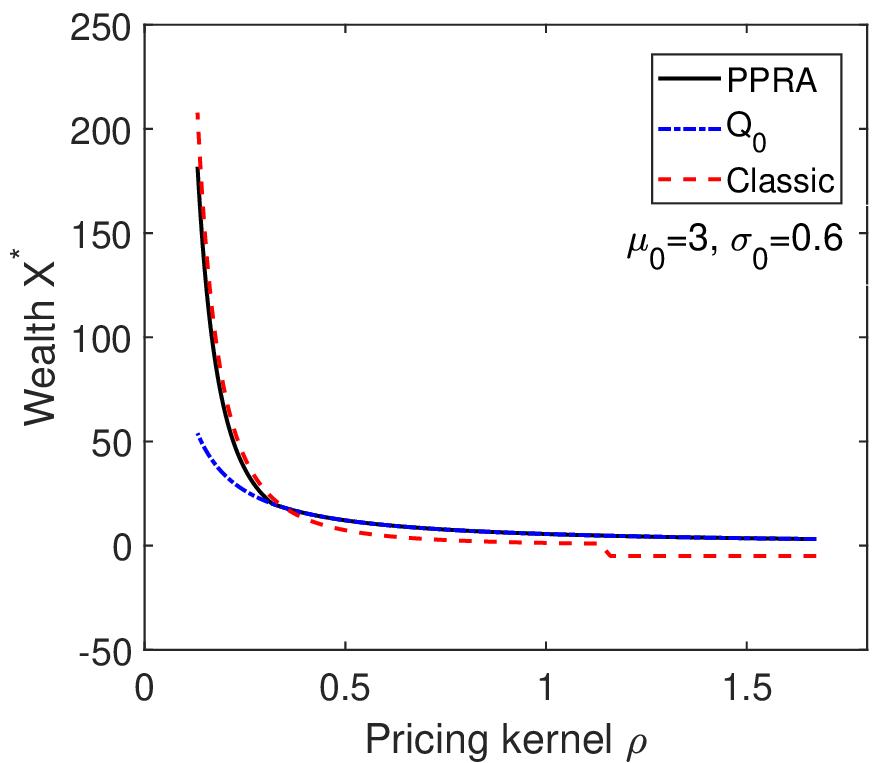}
		\end{matrix}
	\end{gather*}
	\begin{gather*}
		\tiny
		\begin{matrix}
			\includegraphics[width=0.42\textwidth, height=0.29\textheight]{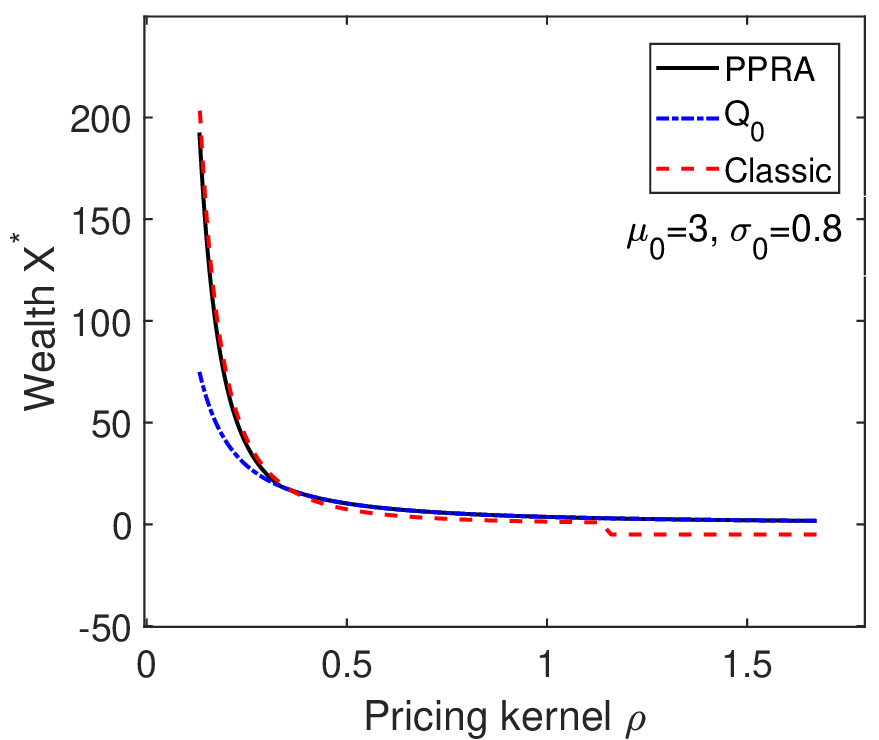}
			& \includegraphics[width=0.42\textwidth, height=0.29\textheight]{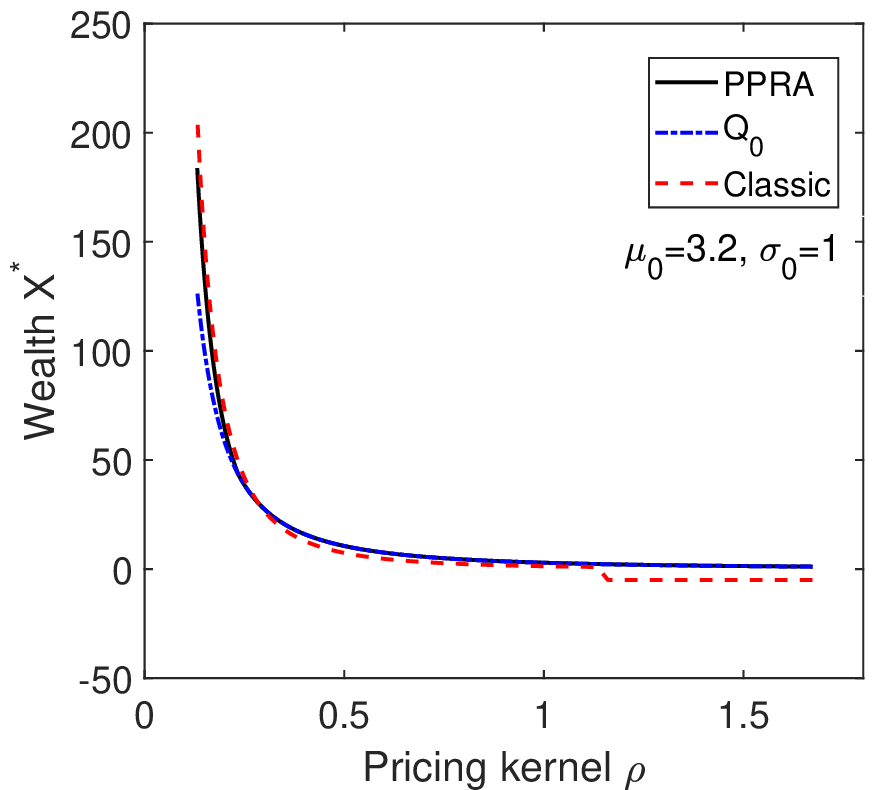}
		\end{matrix}
	\end{gather*}
	\begin{gather*}
		\tiny
		\begin{matrix}
			\includegraphics[width=0.42\textwidth, height=0.29\textheight]{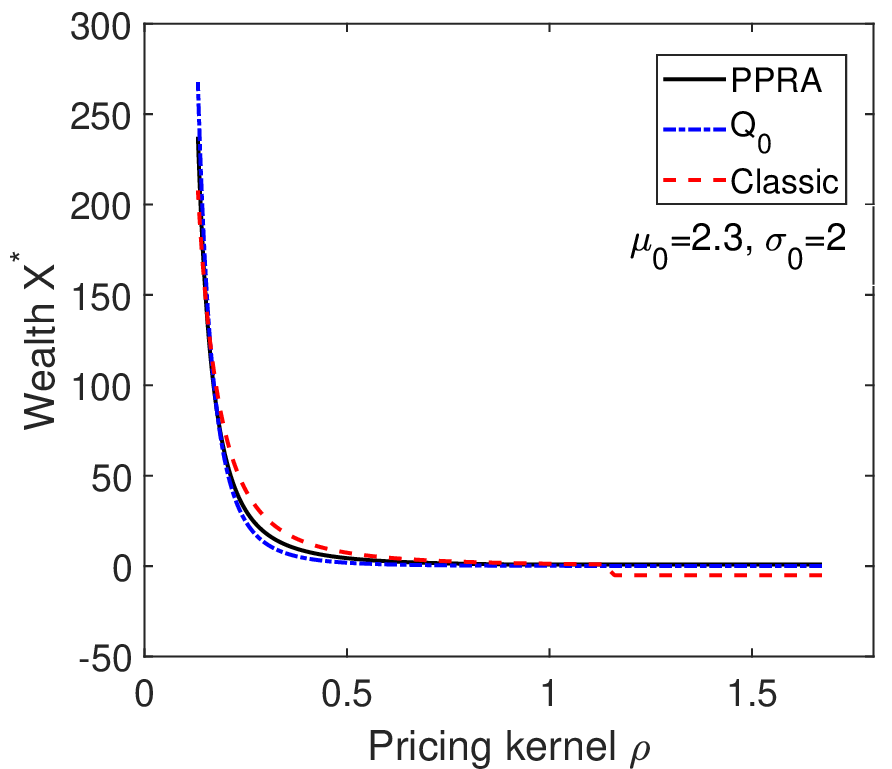}
			& \includegraphics[width=0.42\textwidth, height=0.29\textheight]{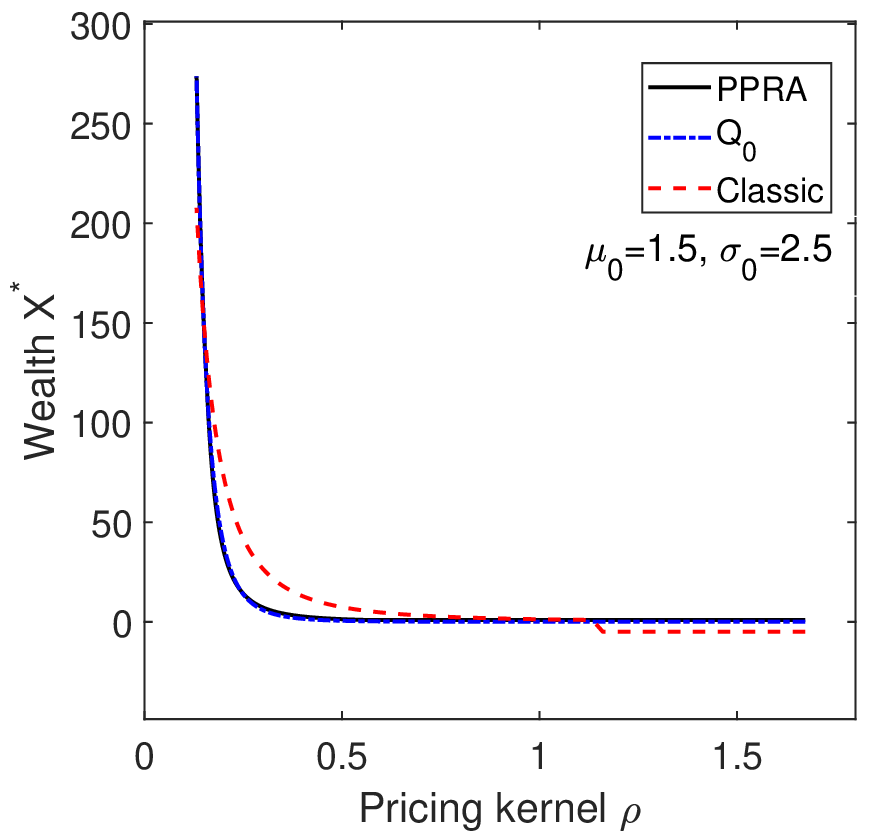}
		\end{matrix}
	\end{gather*}
	\caption{Impacts of $\mu_0$ and $\sigma_0$. }
	\label{fig:SSD_SU}
\end{figure}

\begin{table}[h]
	\centering
	\begin{tabular}{c|c|c|c|c|c|c}
		\hline
		&parametrization & budget of $Q_0$ & poor performance region $C$ & $\lambda$ & $\lambda_\text{cla}$ & partition parameter\\
		\hline
		(a) & $(\mu_0, \sigma_0) = (3, 1)$ & 7.1231 & $(0.6089,1)$  & $0.9105$  & 0.8979 & $t_1 = 1$ \\
		\hline
		(b) & $(\mu_0, \sigma_0) = (3, 0.6)$ & 6.4109 & $(0.4978, 1)$  & 0.9471  & 0.8979 & $t_1 = 1$ \\
		\hline
		(c) & $(\mu_0, \sigma_0) = (3, 0.8)$ & 6.6238 & $(0.5394,1)$ & 0.9255 & 0.8979 & $t_1 = 1$\\
		\hline
		(d) & $(\mu_0, \sigma_0) = (3.2, 1)$ & 8.7002 & (0.2858, 1) & 0.9430 & 0.8979 & $t_1 = 1$\\
		\hline
		(e) & $(\mu_0, \sigma_0) = (2.3, 2)$ & 9.2691 & $(0, 0.4355)\cup (0.9669,1)$ & 1.1987 & 0.8979 & $t_2 = 1, t_1 = 0.0061$\\
		\hline
		(f) & $(\mu_0, \sigma_0) = (1.5, 2.5)$ & 9.8087 & $(0, 0.6840) \cup (0.7726, 1)$ & 2.1508 & 0.8979 & $t_2 = 1, t_1 = 0.0957$\\
		\hline
	\end{tabular}
	\caption{Numerical results in Figure \ref{fig:SSD_SU}.}
	\label{nume:SU}
\end{table}

We apply the Poor-Performance-Region Algorithm, and the explanations and financial insights from Figure \ref{fig:SSD_SU} and Table \ref{nume:SU} are summarized as follows:
\begin{enumerate}
	
	
	
	

    \item Figure \ref{fig:SSD_SU} (a) vs (b) vs (c) vs (d): These cases correspond to scenarios in which the poor performance region consists of a single interval. Numerically, a smaller $\mu_0$ reduces the size of the poor performance region $C$ whereas a smaller $\sigma_0$ enlarges it. 

    \item Figure \ref{fig:SSD_SU}: Across the S-shaped utility cases, we observe that when the market performs poorly, relative to the classical solution, the PPRA wealth exhibits a clear improvement, driven by the SSD constraint. Therefore, in an adverse market scenario, the SSD constraint effectively performs as a risk-control mechanism, ensuring that the PPRA wealth remains at least as high as the benchmark wealth.
    
	\item Figure \ref{fig:SSD_SU} (e) vs (f): In these examples, the poor performance region splits into two disjoint intervals. We further observe that when the budget of $Q_0$ is closer to the bound $\bar{x}$, the PPRA wealth becomes closer to the benchmark. This phenomenon arises because a high level of the budget of $Q_0$ makes the SSD constraint dominate the optimization.

    \item Figure \ref{fig:SSD_SU} (a) vs (b) vs (d): A deeper examination of the impact of the budget of $Q_0$ shows a consistent pattern: as the budget of $Q_0$ approaches the bar $\bar{x}$, the PPRA wealth converges towards the benchmark. This illustrates how the budget level critically shapes the structure of the PPRA wealth.
	
\end{enumerate}

\subsection{SSD Problem: Various Utilities and Benchmark Quantiles}\label{subsection:ssd_numerical}
Based on the proposed algorithm, we further extend its applicability to a broader class of the SSD problems. 
To assess the generality of our approach, we conduct numerical experiments using different utilities and benchmark quantiles $Q_0$.

For the utility function $U(x)$, we consider several representative forms capturing different risk preferences, as shown in Table \ref{tab:utility_cases}.
\begin{table}[H]
\centering
\begin{tabular}{c|c|c|c}
\hline
& exponential & log & piecewise \\
\hline
$U(\cdot)$&$-\dfrac{1}{p} \exp(-p \cdot x), \;x>0$ & $\log(x), \;x>0$ & $\begin{cases}
(x - 1)^{p_1},                   & x \geq 2 \\
-\lambda_1 (1 - x)^{q_1},         & 1 \leq x < 2 \\
x^{p_2} + C,                      & 0 \leq x < 1 \\
C - \lambda_2 (-x)^{q_2},          & -1 \leq x < 0
\end{cases}$ \\
\hline
\end{tabular}
\caption{Settings of various utilities.}
\label{tab:utility_cases}
\end{table}


For the benchmark quantile $Q_0$, we consider the following four cases, as shown in Table \ref{tab:quantile_functions}. 
\begin{table}[H]
\centering
\begin{tabular}{c|c|c|c|c}
\hline
&exponential & log-normal & normal & uniform \\
\hline
$Q_0(\cdot)$&$-\dfrac{\log(1-t)}{\alpha} + k_0$ &
$\exp\left(\sigma_0 \cdot \Phi^{-1}(t) + \mu_0\right) + k_0$ &
$\sigma_0 \cdot \Phi^{-1}(t) + \mu_0$ &
$k t + k_0$ \\
\hline
\end{tabular}
\caption{Benchmark quantiles $Q_0(t)$, $t \in (0,1)$.}
\label{tab:quantile_functions}
\end{table}



The general settings coincide with the previous numerical examples and are given in Table \ref{numerical:general set}.
\begin{table}[H]
\centering
\begin{tabular}{c c c c c c c c c c c c c c c}
    \hline
    $r$ & $\mu_{\text{S}} $ & $\sigma_\text{S}$ & $\theta$ & $T$ & $\mu$ & $\sigma$  \\
    \hline
    0.05 & 0.086 & 0.3 & 0.12 & 20 & -1.1440 & 0.5367 \\
    \hline
\end{tabular}
\caption{General settings in Figures \ref{compose_1}-\ref{compose_2}.}
\label{numerical:general set}
\end{table}

We investigate several combinations of the utilities and the benchmark quantiles. To illustrate the effectiveness of our proposed algorithm, we present the most representative cases under the parameter settings in Table \ref{setup_numerical}.

\begin{table}[H]
\centering
\begin{tabular}{c|c|c|c|c|c}
\hline
& utility & distribution of $Q_0$ & $\overline{x}$ & parameters of $Q_0$ & parameters of utility\\
\hline
(a) & exponential & uniform  & 0.3 & $k=1, k_0=0$ & $p=0.6$\\
\hline
(b) & exponential & exponential  & 0.3 & $\alpha=1.5, k_0=0$ & $p=0.6$\\
\hline
(c) & log & normal  & 1.8 & $\mu_0=5, \sigma_0=1$ & --\\
\hline
(d) & log & uniform  & 1.4 & $k=10, k_0=0$ & --\\
\hline
(e) & piecewise & log-normal & 3.5 & $\mu_0 = -1, \sigma_0 = 3, k_0 = 2.3 $& \makecell{$p_1 = q_1 = 0.6, p_2 = 0.8,$ \\$ q_2 = 0.9, \lambda_1 = 1, \lambda_2 = 2$} \\
\hline
(f) & piecewise & exponential & 1.3 & $\alpha
 = 0.7, k_0 = 2.3 $& \makecell{$p_1 = q_1 = 0.6, p_2 = 0.8,$ \\$ q_2 = 0.9, \lambda_1 = 1, \lambda_2 = 2$} \\
\hline
\end{tabular}
\caption{Parameter settings in Figures \ref{compose_1}-\ref{compose_2}.}
\label{setup_numerical}
\end{table}

Applying the Algorithm \ref{alg}, we obtain the numerical results in Table \ref{result_numerical} and plot Figures \ref{compose_1}-\ref{compose_2}, where each panel illustrates the structure of the correction function $y_{\text{sub}}(\cdot)$ and the PPRA solution.

\begin{table}[H]
\centering
\begin{tabular}{c|c|c|c|c}
\hline
 & poor performance region $C$ & partition parameter & $\lambda$ & $\lambda_{\text{cla}}$\\
\hline
(a) & $(0.8803,1)$ & $ t_1 = 1$ & $\lambda = 1.5540$ & $\lambda_{\text{cla}} = 1.4429$\\
\hline
 (b) & $(0.8904,1)$ & $ t_1 = 1$ & $\lambda = 1.5498$ & $\lambda_{\text{cla}} =1.4429$\\
\hline
 (c) & $(0.4608,1)$ & $ t_1 = 1$ & $\lambda = 0.6497$ & $\lambda_{\text{cla}}= 0.5556$\\
\hline
 (d) & $(0.0426, 0.7236)$ & $ t_1 = 0.1766$ & $\lambda = 0.8260$ & $\lambda_{\text{cla}} = 0.7143$\\
\hline
 (e) & $(0,0.2419)\cup(0.8930,1)$ & $ t_2 = 1,t_1 = 0.0049$ & $\lambda = 1.7002$ & $\lambda_{\text{cla}} =0.9005$\\
\hline
 (f) & $(0.3902,1)$ & $ t_1 = 1$ & $\lambda = 1.7930$ & $\lambda_{\text{cla}} = 1.5516$\\
\hline
\end{tabular}
\caption{Numerical results in Figures \ref{compose_1}-\ref{compose_2}.}
\label{result_numerical}
\end{table}

\begin{figure}[htbp]
    \begin{gather*}
        \tiny
        \begin{matrix}
            \includegraphics[width=0.45\textwidth,height=0.403\textwidth]{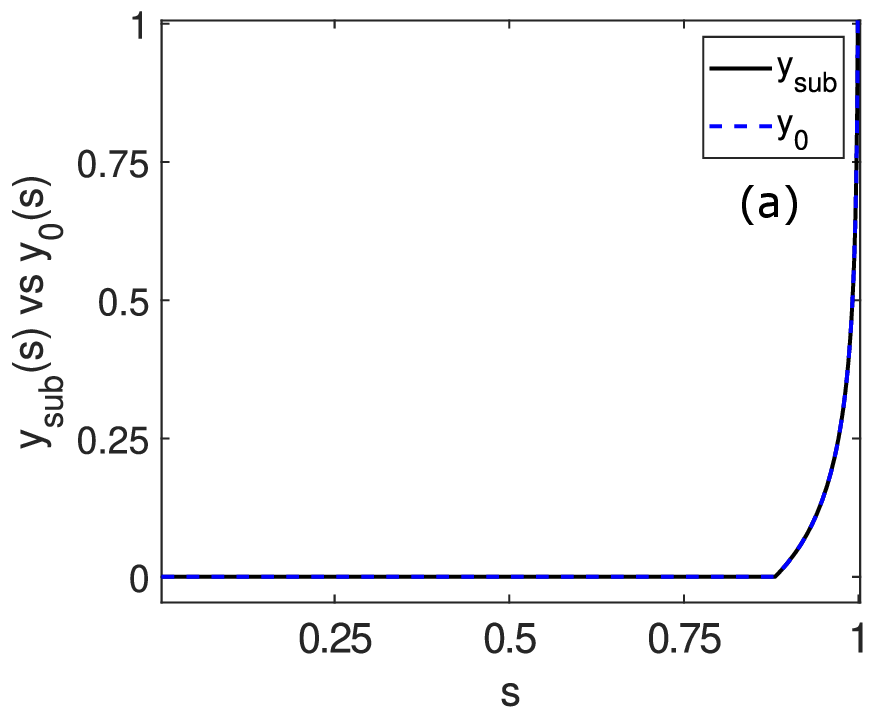}
            & \includegraphics[width=0.45\textwidth,height=0.41\textwidth]{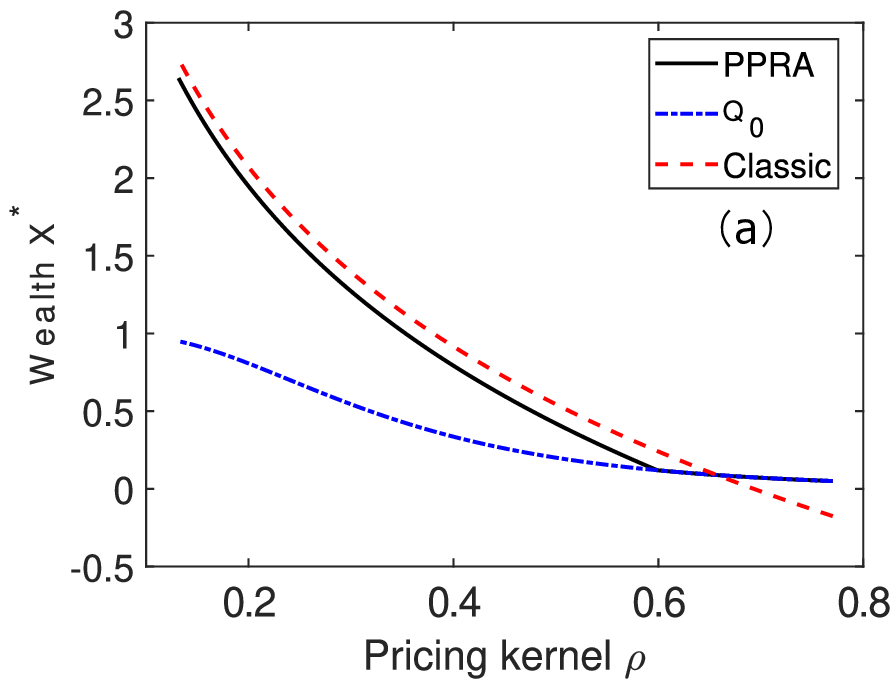}
        \end{matrix}
    \end{gather*}
    \begin{gather*}
        \tiny
        \begin{matrix}
            \includegraphics[width=0.45\textwidth,height=0.403\textwidth]{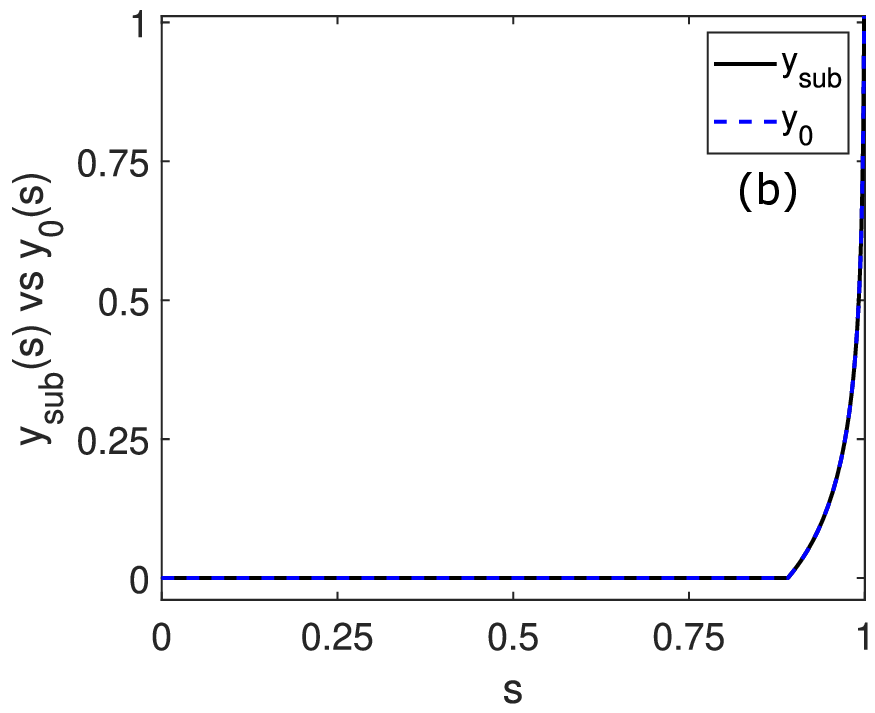}
            & \includegraphics[width=0.45\textwidth,height=0.41\textwidth]{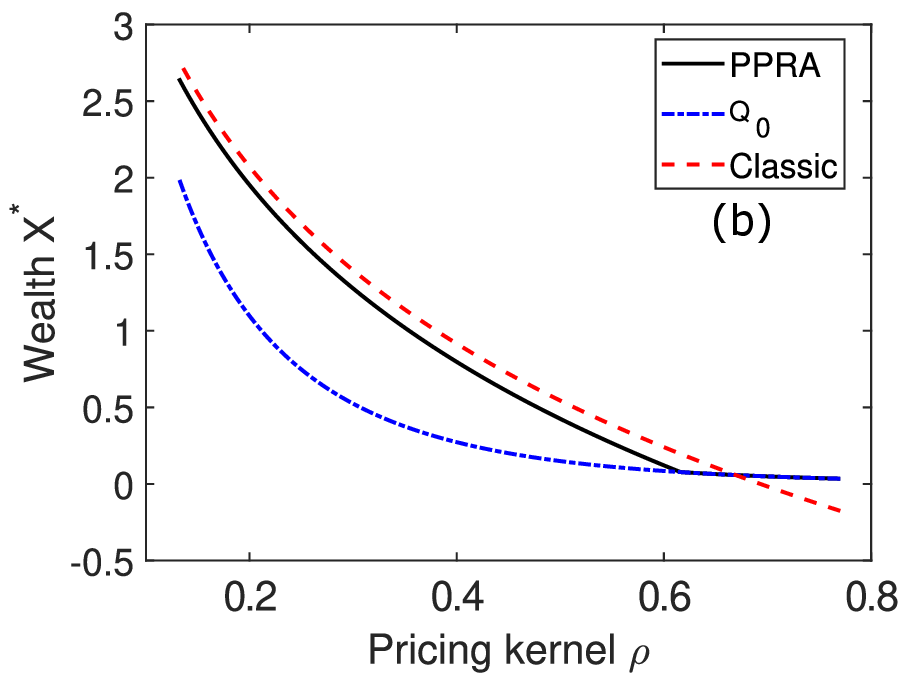}
        \end{matrix}
    \end{gather*}
    \begin{gather*}
        \tiny
        \begin{matrix}
            \includegraphics[width=0.45\textwidth,height=0.403\textwidth]{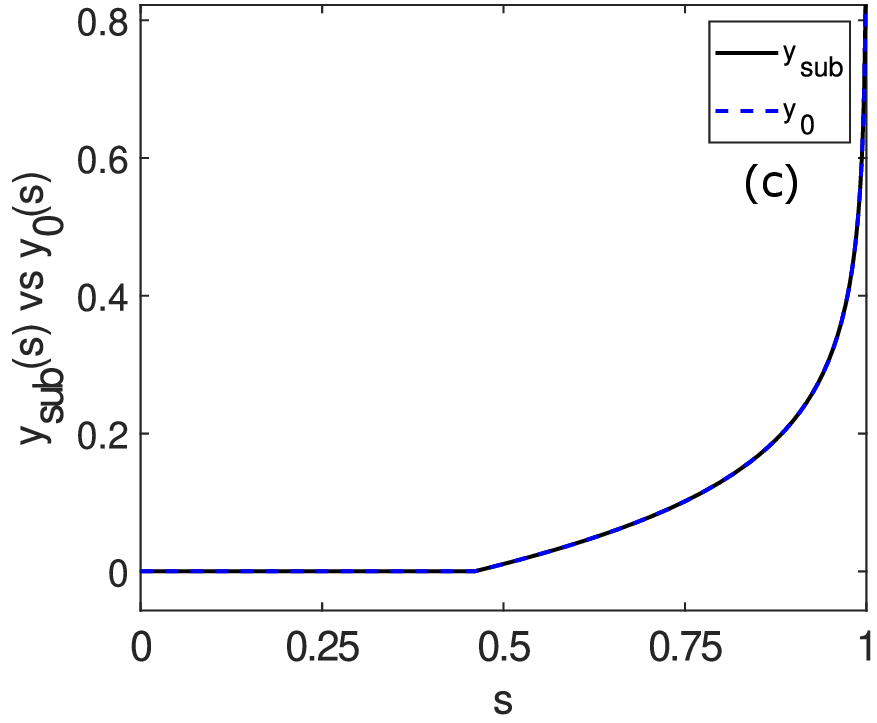}
            & \includegraphics[width=0.45\textwidth,height=0.41\textwidth]{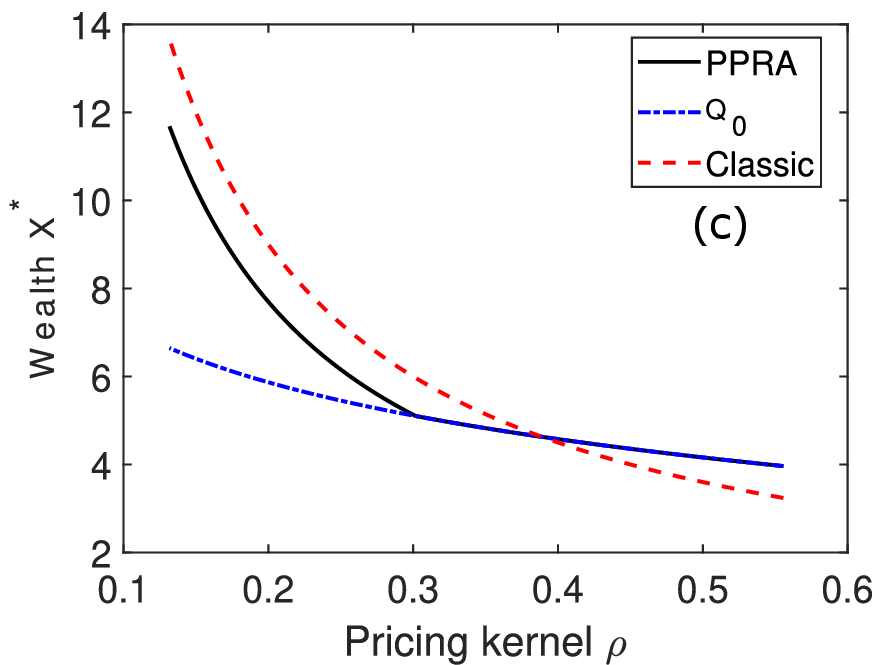}
        \end{matrix}
    \end{gather*}
\caption{Settings (a)(b)(c) in Table \ref{result_numerical}.}
\label{compose_1}
\end{figure}

\begin{figure}[htbp]
    \begin{gather*}
        \tiny
        \begin{matrix}
            \includegraphics[width=0.45\textwidth,height=0.403\textwidth]{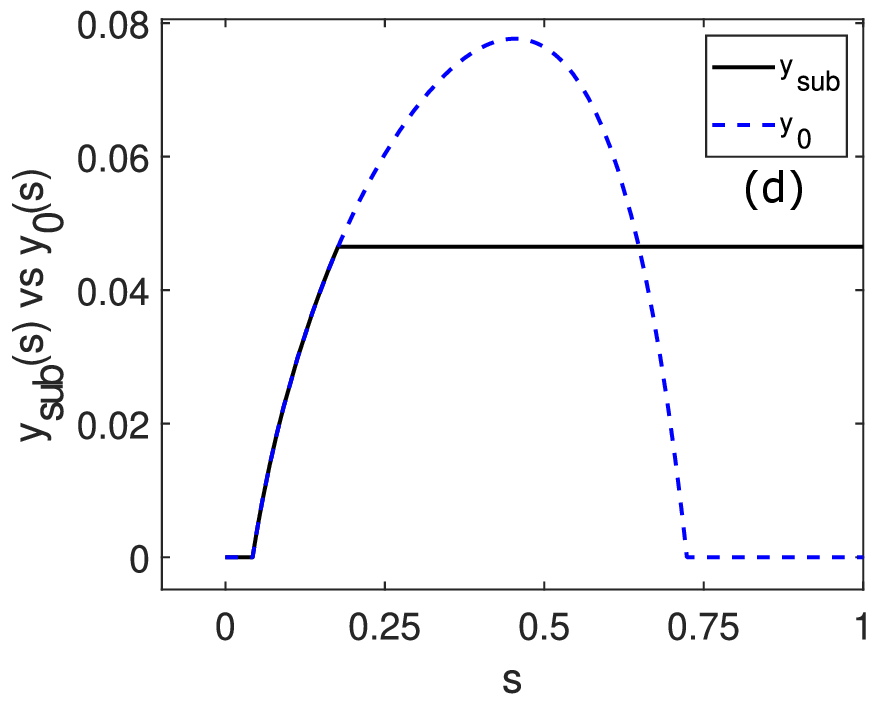}
            & \includegraphics[width=0.45\textwidth,height=0.41\textwidth]{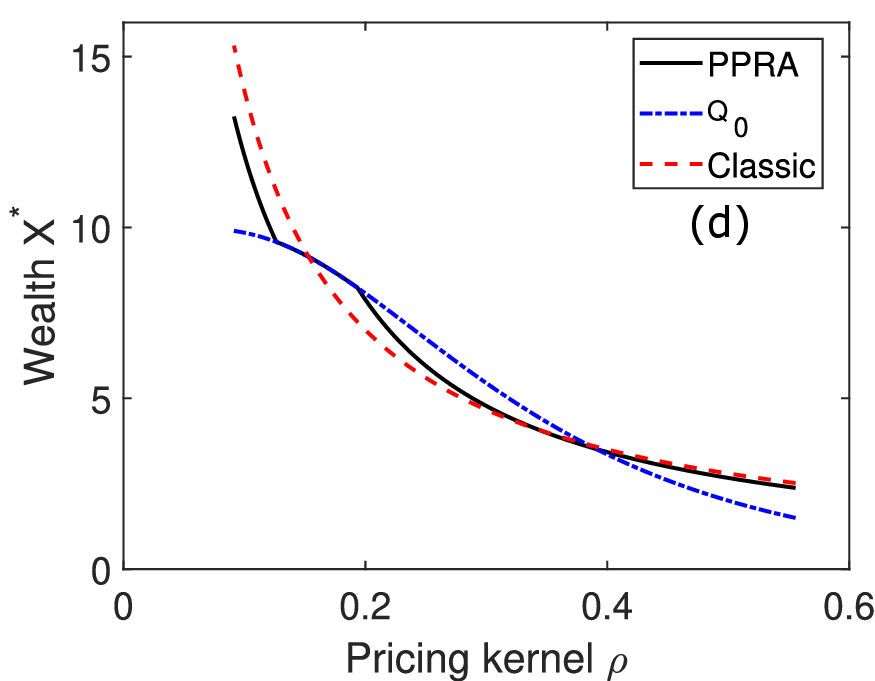}
        \end{matrix}
    \end{gather*}
    \begin{gather*}
        \tiny
        \begin{matrix}
            \includegraphics[width=0.45\textwidth,height=0.403\textwidth]{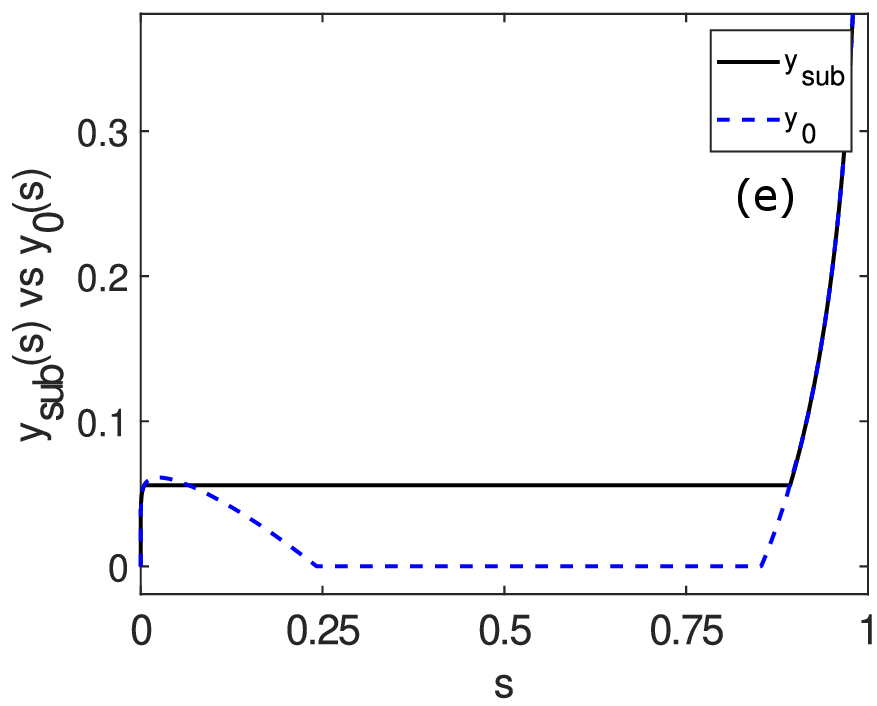}
            & \includegraphics[width=0.45\textwidth,height=0.41\textwidth]{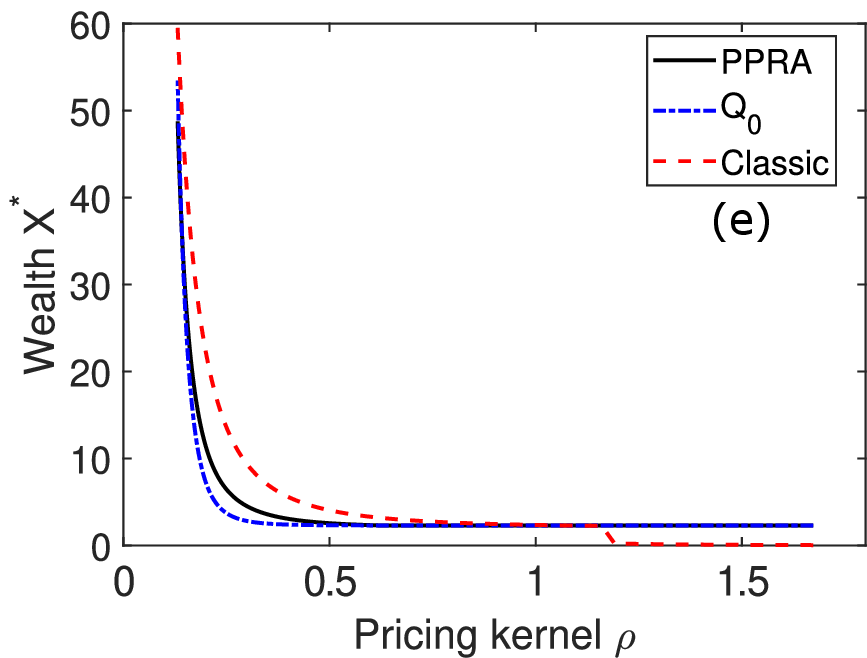}
        \end{matrix}
    \end{gather*}
    \begin{gather*}
        \tiny
        \begin{matrix}
            \includegraphics[width=0.45\textwidth,height=0.403\textwidth]{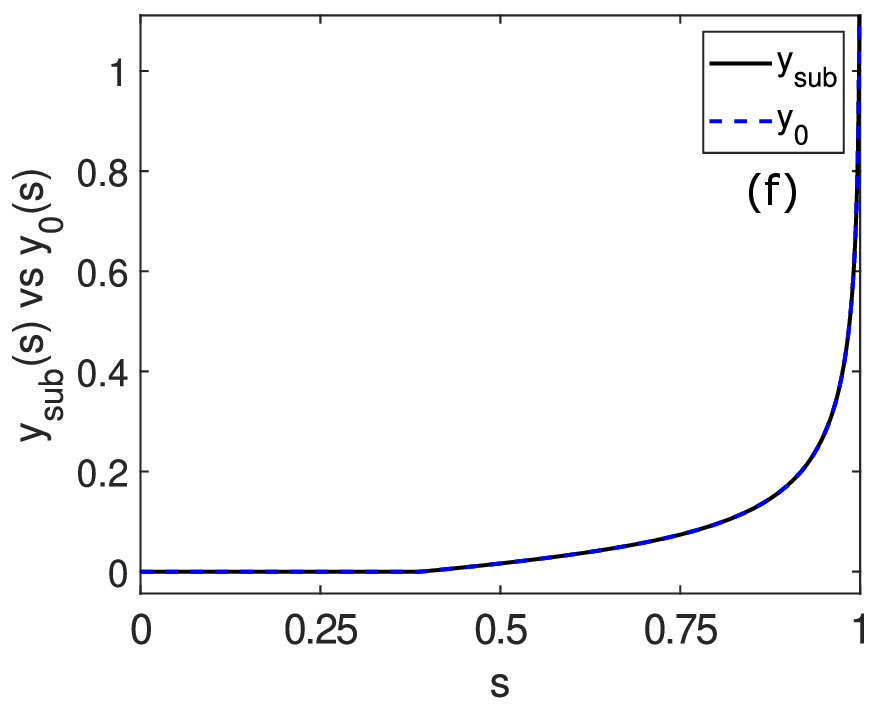}
            & \includegraphics[width=0.45\textwidth,height=0.41\textwidth]{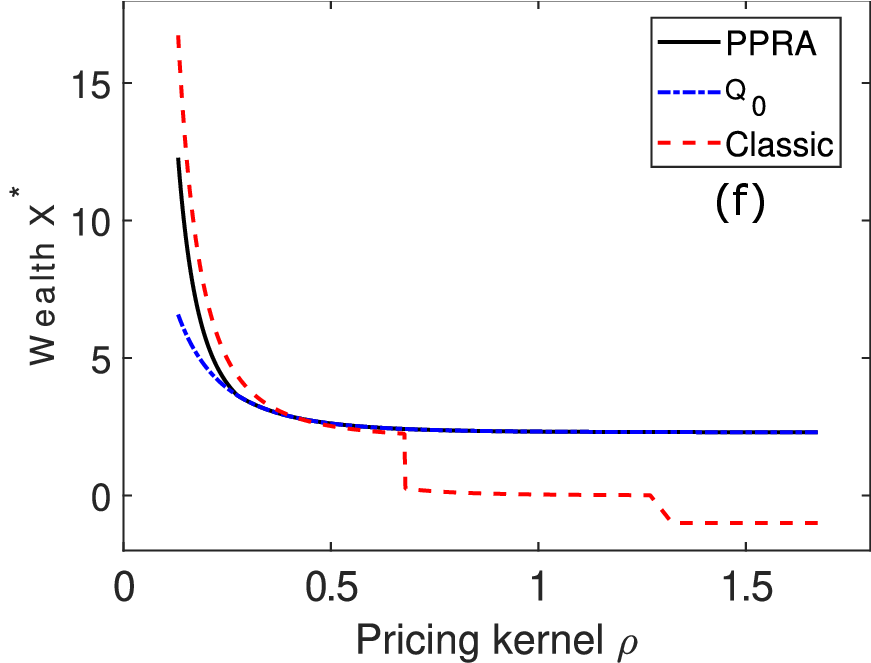}
        \end{matrix}
    \end{gather*}
\caption{Settings (d)(e)(f) in Table \ref{result_numerical}.}
\label{compose_2}
\end{figure}

Later, we will clarify the structural behavior of the correction function $y_{\text{sub}}(\cdot)$. The numerical experiments yield the following explanations and financial insights:
\begin{enumerate}
    \item Figure \ref{compose_1} (a)(b)(c): In a bad market scenario, the correction function $y_{\text{sub}}(\cdot)$ satisfies $y_{\text{sub}}(\cdot) = y_0(\cdot)$. Therefore, the PPRA wealth will coincide with the benchmark wealth. In this scenario, the benchmark serves as a risk-control mechanism, which is similar to Figure \ref{fig:SSD_power} (a)(b)(d).

    \item Figure \ref{compose_1} (c)(f): Compared with Figure \ref{compose_1} (a)(b), the region where $y_{\text{sub}}(\cdot) = y_0(\cdot)$ expands. Hence, under relatively stagnant or bad market scenario, the PPRA wealth tends to coincide with the benchmark wealth. In such scenario, the SSD constraint plays a more prominent role by ensuring a benchmark wealth to reduce risk.    

    \item Figure \ref{compose_2} (d): The correction function $y_{\text{sub}}(\cdot) $ equals $ y_0(\cdot)$ over certain intermediate intervals, causing the PPRA wealth to coincide with the benchmark in the mid-range market scenario, as shown in (d)'s right panel.

    \item Figure \ref{compose_2} (e): The correction function $y_{\text{sub}}(\cdot)$ equals $y_0(\cdot)$ at both endpoints of the interval $(0,1)$. This implies the PPRA wealth must coincide with the benchmark wealth in both extremely bad and exceptionally favorable market scenarios.
\end{enumerate}

In Figure \ref{compose_2} (e), we can observe how Steps 8-12 of Algorithm~\ref{alg} ensure that the constructed function $y_{\text{sub}}(\cdot)$ remains non-decreasing on $(0,1)$. The key idea is to modify the construction of $g_i(t)$ in Eq.~\eqref{eq:gi} so that both monotonicity requirements of $y_{\text{sub}}(\cdot)$ and the condition $z^*(\cdot) \le 0$ are satisfied throughout the region in which $y_{\text{sub}}(\cdot)$ is updated. In this example, the interval where $y_{\text{sub}}(\cdot) = y_{\text{sub}}(t_1)$ is extended into the region covered by the previous iteration (i.e., the first iteration), updating the original $y_{\text{sub}}(\cdot)$ in this region.

Across various combinations of the utilities and the benchmark quantiles $Q_0$, we observe that the poor-performance region typically consists of one or two disjoint intervals. The proposed algorithm exhibits strong adaptability under these diverse settings, effectively identifying the sub-optimal solution in most cases. At the same time, our algorithm is able to handle certain special cases in which the original piecewise construction fails to guarantee monotonicity. By modifying the construction of $g_i(t)$, the algorithm provides a valid sub-optimal solution and, as a result, restores monotonicity and extends the applicability of our approach to more complex configurations.

\section{Algorithm-Guided Piecewise-Neural-Network Framework}
\label{sec:NN}


In this section, we propose a novel approach to solve the SSD Problem \eqref{prob:SSD_qf} by designing a piecewise-neural-network-framework combined with Algorithm \ref{alg}. The key observation is that in Algorithm \ref{alg}, the poor performance region $C$ defined in Eq. \eqref{eq:relation_C_y0} and the construction of the correction function $y_{\text{sub}}(\cdot)$ provide valuable structure information about the optimal solution, thereby guiding the design of an effective neural network framework.

We begin by applying Algorithm \ref{alg}, which yields a correction function $y_{\text{sub}}(\cdot)$. Using $y_{\text{sub}}(\cdot)$, we derive a sub-optimal solution $Q_{\text{sub}}(\cdot)$ as presented in Eq. \eqref{eq:Q_algo}. 
To design a neural network framework, the key idea is that we use the structure information of the derived $Q_{\text{sub}}(\cdot)$ to build a piecewise-neural-network framework. Then we apply this framework to train a solution for Problem \eqref{prob:SSD_qf}.

\subsection{Model Settings}
\begin{algorithm}[htbp]\small
\caption{Algorithm-guided piecewise-neural-network framework for SSD Problem \eqref{nn:goal}}\label{alg_nn}
\begin{algorithmic}[1]
    \State Implement Algorithm \ref{alg} and get the sub-optimal solution structure as Eq. \eqref{eq:Q_algo}. Using the partition intervals $(s_k, s_{k+1}]$ defined in Eq. \eqref{eq:nn_Qsub}, initialize the neural network architecture as follows:
    \begin{equation}\label{eq:Qtheta}
    \begin{aligned}
    Q_\theta(s) &=
    \begin{cases}
    f^{(0)}_{\theta_0}(\mathbf{s}_{\text{feat}}(s)), & s \in (s_0, s_1], \\
    f^{(1)}_{\theta_1}(\mathbf{s}_{\text{feat}}(s)), & s \in (s_1, s_2], \\
    \;\vdots & \\
    f^{(K)}_{\theta_K}(\mathbf{s}_{\text{feat}}(s)), & s \in (s_K, s_{K+1}),
    \end{cases}
    \end{aligned}
    \end{equation}
    where $f^{(k)}_{\theta_k}$ denotes the $k$-th neural sub-network parameterized by $\theta_k$, the parameters of the integrated network $Q$ are denoted as $\theta = \{\theta_0, \theta_1, \dots, \theta_K\}$ and  $\mathbf{s}_{\text{feat}}(s)$ represents the Fourier feature embedding of the scalar input $s$ defined as
    \begin{equation}\label{eq:s_feat}
        \mathbf{s}_{\text{feat}}(s) = [\,\sin(2\pi s),\; \sin(4\pi s),\; \cos(2\pi s),\; \cos(4\pi s)\,]^\top \in \mathbb{R}^4.
    \end{equation}

    \State Define a prior function as follows:
    \begin{equation}\label{eq:prior}
    \begin{aligned}
        \phi(s) =
        \begin{cases}
            I(Q_\rho(1-s)), & s \in (s_0, s_1],\\
            Q_0(s), & s \in (s_1, s_2],\\
            \;\vdots & \\
            I(Q_\rho(1-s)), & s\in (s_K, s_{K+1}).
        \end{cases}
    \end{aligned}
    \end{equation}

    \State For each interval, update the sub-network output with the analytic prior term and activate the integrated network $Q_\theta(\cdot)$ as follows:
    $$
    Q_\theta(s) \leftarrow Q_\theta(s) + \phi(s), \;\; Q_\theta(s) \leftarrow \text{ReLU}(Q_\theta(s)).
    $$

    \State Define the objective function as the expected utility from Eq. \eqref{nn:goal}.
    Since the integral cannot be computed analytically, we approximate it by uniformly sampling $n$ points $s_i \in (0,1)$ and compute
    \begin{equation}\label{nn:obj}
        \mathcal{L}_{\text{obj}}(\theta) = \frac{1}{n} \sum_{i=1}^{n} U\big(Q_\theta(s_i)\big).
    \end{equation}

    \State Take the budget constraint ($C_1$) and SSD constraint ($C_2$) as the penalty of the loss function. Define
    $$
    C_1 = \frac{1}{n} \sum_{i=1}^n Q_\theta(s_{i}) Q_\rho(1 - s_i), \;\; C_2 = \max\left\{0, \max_{k=1,\dots,n} \left[ \frac{1}{k}\sum_{i=1}^k Q_0(s_{i}) - \frac{1}{k}\sum_{i=1}^k Q_\theta(s_i) \right] \right\}.
    $$

    \State Add weights $w_1, w_2$ to $C_1, C_2$ and calculate the loss function as follows
    \begin{equation}\label{eq:penalty}
    \mathcal{L}_{\text{p1}} = w_1 \, (C_1 - \bar{x})^2, \;\;\mathcal{L}_{\text{p2}} = w_2 \cdot C_2,
    \end{equation}

    $$
    \mathcal{L}_{\text{total}}(\theta) = -\mathcal{L}_{\text{obj}}(\theta) + \mathcal{L}_{\text{p1}} + \mathcal{L}_{\text{p2}}.
    $$

    \State Next, start the training process.
    \Require Neural sub-networks $\{f_{\theta_k}^{(k)}\}_{k=0}^K$, pricing kernel function $Q_\rho(\cdot)$, benchmark function $Q_0(\cdot)$, utility function $U(\cdot)$, budget $\bar{x}$, sample size $n$, learning rate $\eta$, number of Adam steps $A_s$, penalty weights $w_1, w_2$.
    \Ensure Trained network $Q_\theta(s)$.
    
    \State Sample $s_i \in (0,1)$, $i=1,\dots,n$.
    \State Compute Fourier features $\mathbf{s}_{\text{feat}}$ as Eq. \eqref{eq:s_feat}.
    
    \For{$k = 0$ to $K$}
        \State Initialize each sub-network: $f_{\theta_k}^{(k)}(\mathbf{s}_{\text{feat}}(s_i))$.
        \State Compute analytic prior $\phi(s_i)$ as Eq. \eqref{eq:prior}.
    \EndFor

    \State Merge the sub-networks to an integrated network $Q_\theta(s)$ as Eq. \eqref{eq:Qtheta}.

    \State Activate the network: $Q_\theta(s) \leftarrow \text{ReLU}(Q_\theta(s))$.

    \For{$i = 0$ to $A_s$}
        \State Compute the objective function: $    \mathcal{L}_{\text{obj}}(\theta) = \frac{1}{n} \sum_{i=1}^{n} U\big(Q_\theta(s_i)\big)$.
        
        \State Compute constraint penalties $\mathcal{L}_{\text{p1}}$ and $\mathcal{L}_{\text{p2}}$ as Eq. \eqref{eq:penalty}.

        \State Compute the total loss in current iteration: $    \mathcal{L}_{\text{total}}(\theta) = -\mathcal{L}_{\text{obj}}(\theta) + \mathcal{L}_{\text{p1}} + \mathcal{L}_{\text{p2}}$.
        
        \State \textbf{Adam update:} $\theta \gets \theta - \eta \nabla_\theta \mathcal{L}_{\text{total}}(\theta)$.
    \EndFor
    
    \State \Return Trained network $Q_\theta(s)$.

\end{algorithmic}
\end{algorithm}
The optimization problem to be approximated by the neural network is formulated as follows:
\begin{equation}\label{nn:goal}
	\begin{aligned}
		& \max_{Q \in \Q_2(Q_0)} \int_{0}^{1} U(Q(s)) \d s ~~ \text{s.t. } \int_{0}^{1} Q(s) Q_\rho(1-s) \d s \leq \overline{x},
	\end{aligned}
\end{equation}
which is identical to Problem (\ref{prob:SSD}).

Based on Algorithm \ref{alg}, we obtain the suboptimal solution $Q_{\text{sub}}(\cdot)$ as given in Eq. (\ref{eq:Q_algo}), which exhibits a piecewise structure. Specifically, $Q_{\text{sub}}(\cdot)$ can be expressed as

\begin{equation}\label{eq:nn_Qsub}
\begin{aligned}
Q_{\text{sub}}(s) &=
\begin{cases}
I(\lambda(Q_\rho(1-s) - y(\cdot)), & s \in (s_0, s_1], \\
Q_0(s), & s \in (s_1, s_2], \\
\;\vdots & \\
I(\lambda(Q_\rho(1-s) - y(\cdot)), & s \in (s_K, s_{K+1}),
\end{cases}
\end{aligned}
\end{equation}
where $s_0 = 0, s_{K+1} = 1.$


Following this structure, we construct a set of neural sub-networks 
$\{f^{(k)}_{\theta_k}\}_{k=0}^K$ to approximate each interval respectively. The main design considerations are summarized as follows:
\begin{enumerate}[a.]
    \item In Step 1, the piecewise formulation ensures that $Q_\theta(s)$ inherits the partition interval of $Q_{\text{sub}}(\cdot)$, enabling the network to better capture the structural behavior of the optimal solution. Specifically, each sub-network $f^{(k)}_{\theta_k}$ is a fully connected feedforward network with 8 hidden layers, each consisting of 256 neurons with Tanh activations.

    \item In Step 2, we introduce an analytical prior term into each sub-network $f_{\theta_k}^{(k)}$. The design of $\phi(s)$ is inspired by the structure of $Q_{\text{sub}}(\cdot)$, preserving essential information from the solution of Algorithm~\ref{alg} as a prior function and helps the network to capture the intrinsic information of the optimal solution.

    \item In Steps 5-6, we square the violation of the budget constraint to represent its penalty, ensuring the budget stays close to $\overline{x}$.  
    For the SSD constraint, we penalize the maximal violation to strictly enforce its satisfaction. The relative importance of satisfying the constraint during training can be adjusted by tuning its associated weight $w_1, w_2$.

    \item In Steps 8-22, the approximation accuracy of the objective function and constraints can be improved by increasing the number of sampled points $n$. The Adam optimizer is used to update the network parameters, and the number of Adam steps $A_s$ determines how many gradient-based updates are performed during training, thus enhancing the convergence of the network.
\end{enumerate}

\subsection{Experimental Results}
To validate the model, we first conduct experiments under conditions in which the optimal solution $Q^*(\cdot)$ is available. This optimality holds because Algorithm \ref{alg} results numerically follow the characterization in Theorem \ref{thm:SSD}. We use the setup in Section \ref{subsection:ssd_numerical}, which considers both the exponential utility and log utility under a linear SSD constraint with benchmark quantile $Q_0$ (Table \ref{setup_numerical} (a)(c)).





We follow the steps in Algorithm \ref{alg_nn} to get a trained neural network $Q_\theta$, and the results are summarized in Table \ref{nn_result}.

\begin{table}[htbp]
\centering
\begin{tabular}{c|c|c|c|c}
\hline
& utility & distribution of $Q_0$ & optimal value & neural network value\\
\hline
(i) & exponential & uniform  & -0.8965 & -0.8990\\
\hline
(ii) & log & uniform  & 1.4781 & 1.4686 \\
\hline
\end{tabular}
\caption{Objective value for neural network approximation.}
\label{nn_result}
\end{table}

\begin{figure}[htbp]
    \begin{gather*}
        \tiny
        \begin{matrix}
            \includegraphics[width=0.45\textwidth, height=0.31\textheight]{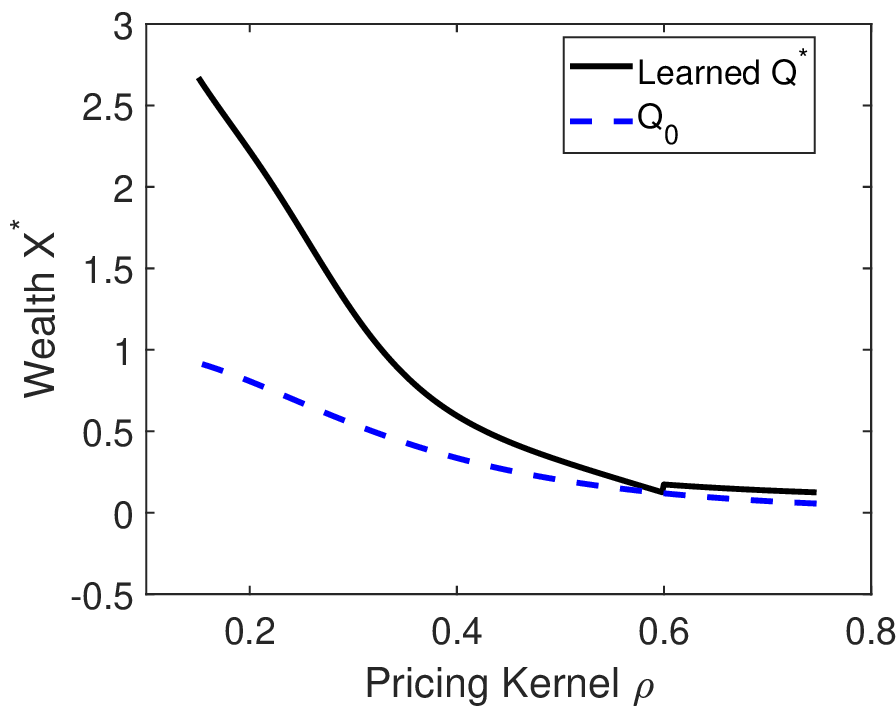}
            & \includegraphics[width=0.45\textwidth,height=0.31\textheight]{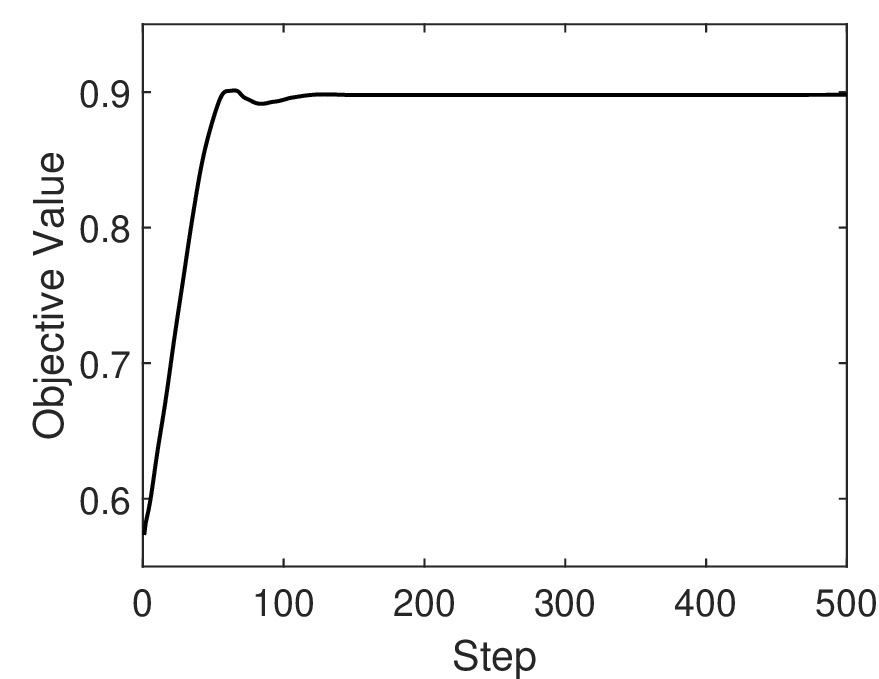}
        \end{matrix}
    \end{gather*}
    \begin{gather*}
        \tiny
        \begin{matrix}
            \includegraphics[width=0.45\textwidth, height=0.31\textheight]{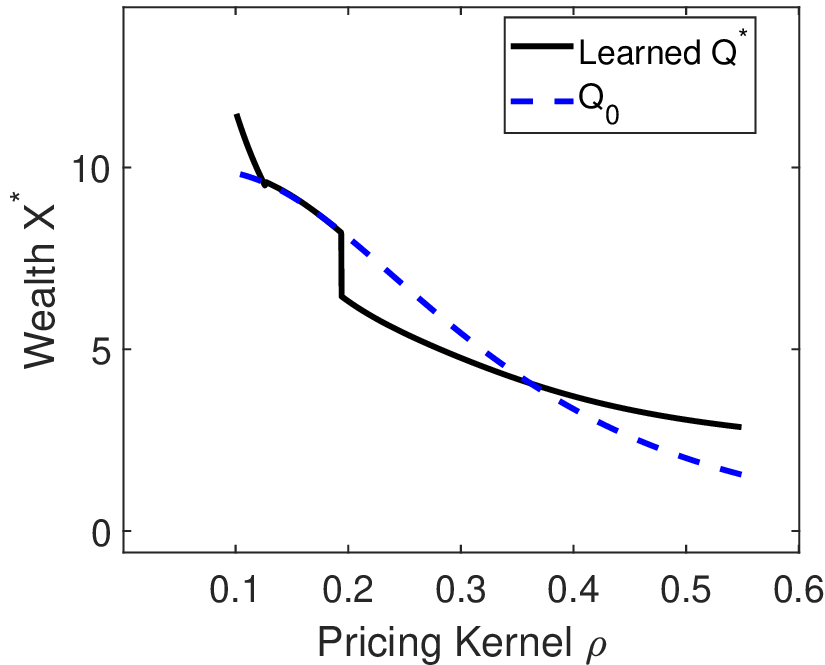}
            & \includegraphics[width=0.45\textwidth,height=0.31\textheight]{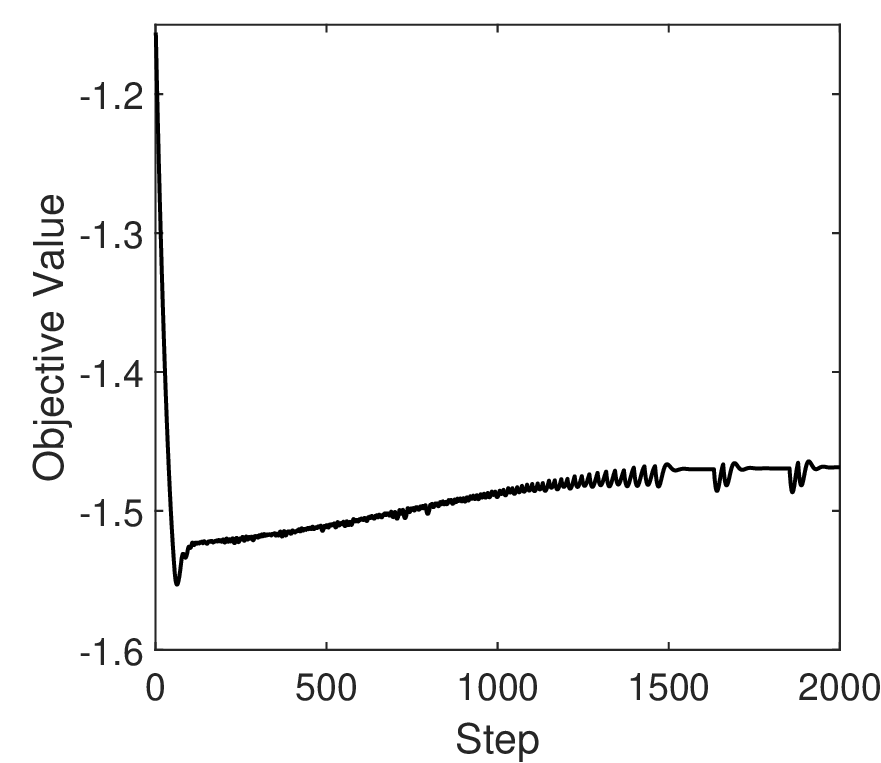}
        \end{matrix}
    \end{gather*}
\caption{Neural network approximation results.}
\label{nn_part1}
\end{figure}

Here are some illustrations of the neural network results:
\begin{enumerate}
    \item Figure \ref{nn_part1} (i)(ii)'s left panel: Compared with Figure \ref{compose_1} (a) and Figure \ref{compose_2} (d), the algorithm-guided piecewise-neural-network framework successfully preserves the structure features of the optimal solution, achieving a close match.

    \item Convergence behavior: The network typically converges within 100-2000 steps. As shown in Table \ref{nn_result}, our algorithm achieves high numerical accuracy relative to the optimal solution, which is reflected in the objective value $\int_{0}^{1} U(Q(s)) \d s$ in Problem \eqref{nn:goal}.
    
    \item Figure \ref{nn_part1} (i)(ii)'s right panel: The convergence speed depends on the number of sub-networks: the exponential case uses 2 sub-networks, while the log case uses 3, which also reflects the problem complexity. In simpler cases (Figure \ref{nn_part1} (i)), the convergence curves are smooth and rapid because the two constraints in Problem \eqref{nn:goal} are quickly fulfilled. In more complex cases (Figure \ref{nn_part1} (ii)), the network spends more time mitigating constraint penalties, resulting in slower convergence.
\end{enumerate}

Here we come to the conclusion that our approach, first deriving a sub-optimal solution via the proposed algorithm and then leveraging it to guide the neural network design, effectively captures the essential properties of the optimal solution and provides stable, high-quality performance for the SSD problem with strictly with concave utilities.


Remarkably, our algorithm-guided piecewise-neural-network framework remains stable under a nonconcave utility and achieves rapid convergence in Figure \ref{nn_part2}. By leveraging sub-optimal solutions obtained from our proposed Algorithm \ref{alg} as analytic priors and structuring the network in a piecewise manner, the framework effectively reduces the functional search space and guides the optimization toward regions that fulfill the SSD constraints. In contrast, a standard monolithic-neural-network (non-piecewise) framework exhibits substantially slower convergence, often requiring tens of times more training steps to reach convergence.

To illustrate this, we consider an S-shaped utility following the setup in Section \ref{subsection:Sshaped} (Table \ref{nume:SU} (b)). Solutions obtained with our Algorithm \ref{alg_nn} are compared against those from a standard monolithic-neural-network framework using identical parameter settings.

\begin{table}[htbp]
\centering
\begin{tabular}{c|c|c|c|c}
\hline
& network & neural network value  & $C_1$ satisfied steps & $C_2$ satisfied steps \\
\hline
(i) & piecewise & 14.7531  & 83 & 10\\
\hline
(ii) & standard & 12.0320 & 10950 & 4279  \\
\hline
\end{tabular}
\caption{Convergence steps and constraint satisfied steps.}
\label{convergenceTable}
\end{table}

\begin{figure}[htbp]
    \begin{gather*}
        \tiny
        \begin{matrix}
            \includegraphics[width=0.45\textwidth, height=0.3\textheight]{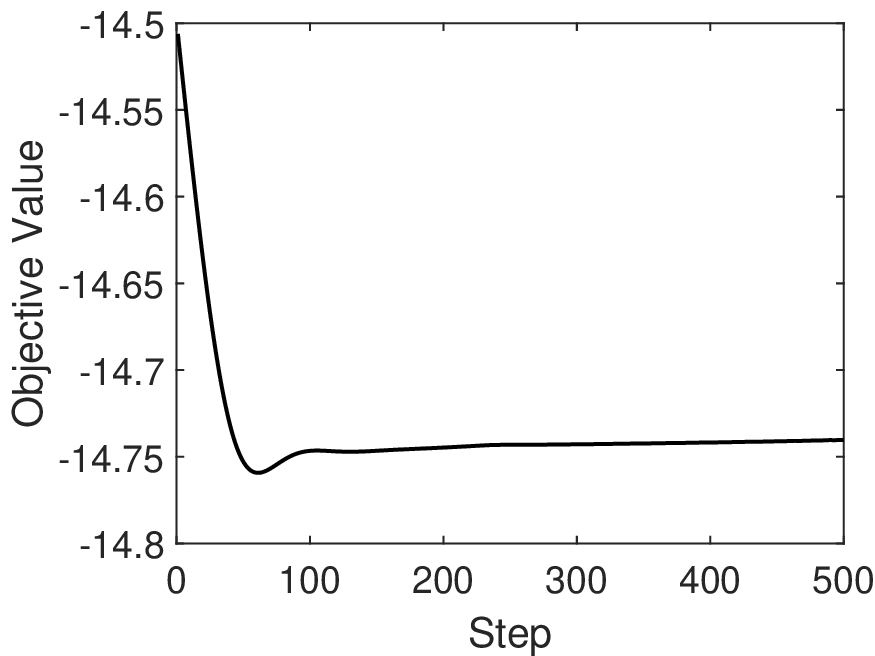}
            & \includegraphics[width=0.45\textwidth,height=0.3\textheight]{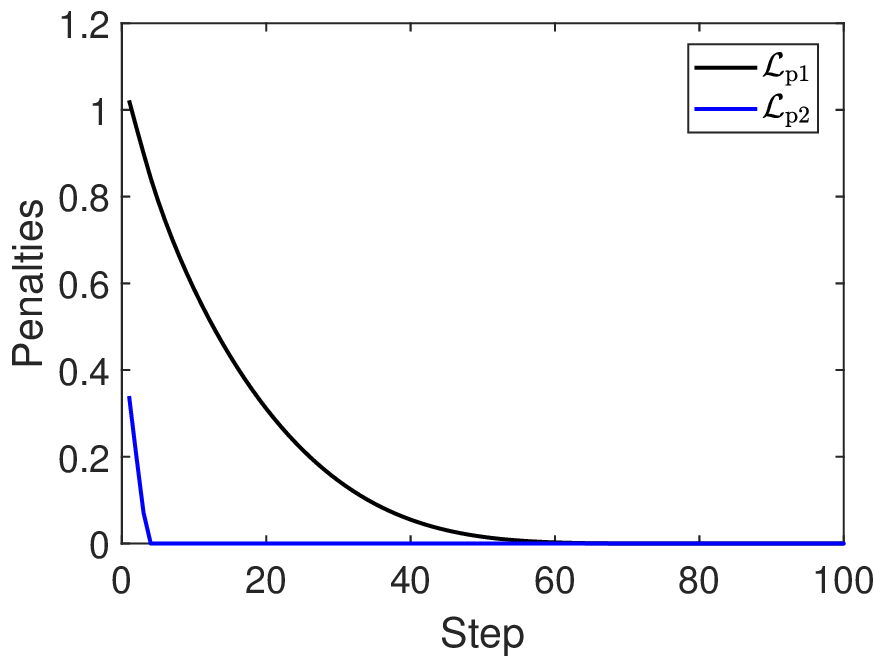}
        \end{matrix}
    \end{gather*}
    \begin{gather*}
        \tiny
        \begin{matrix}
            \includegraphics[width=0.45\textwidth, height=0.3\textheight]{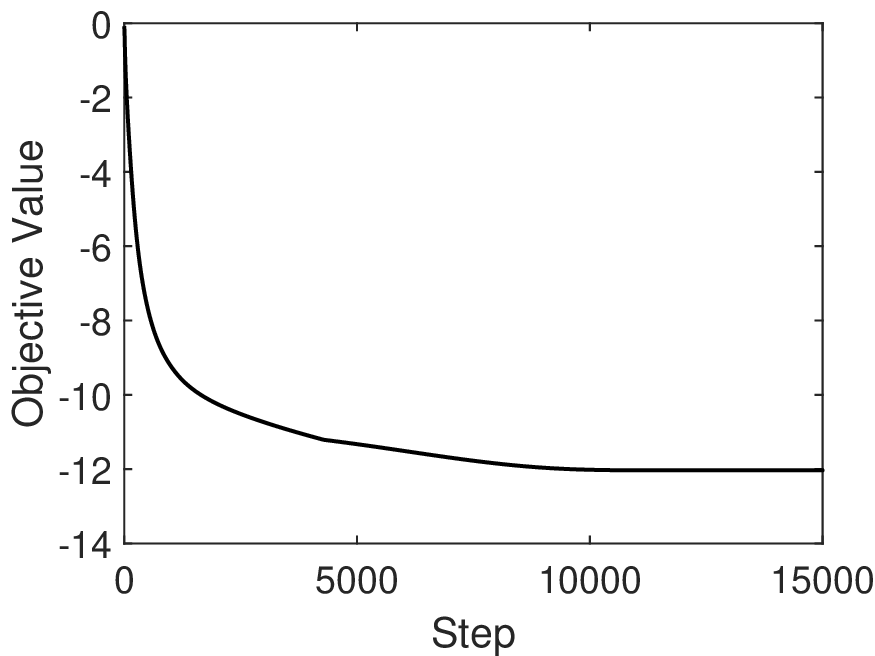}
            & \includegraphics[width=0.45\textwidth,height=0.295\textheight]{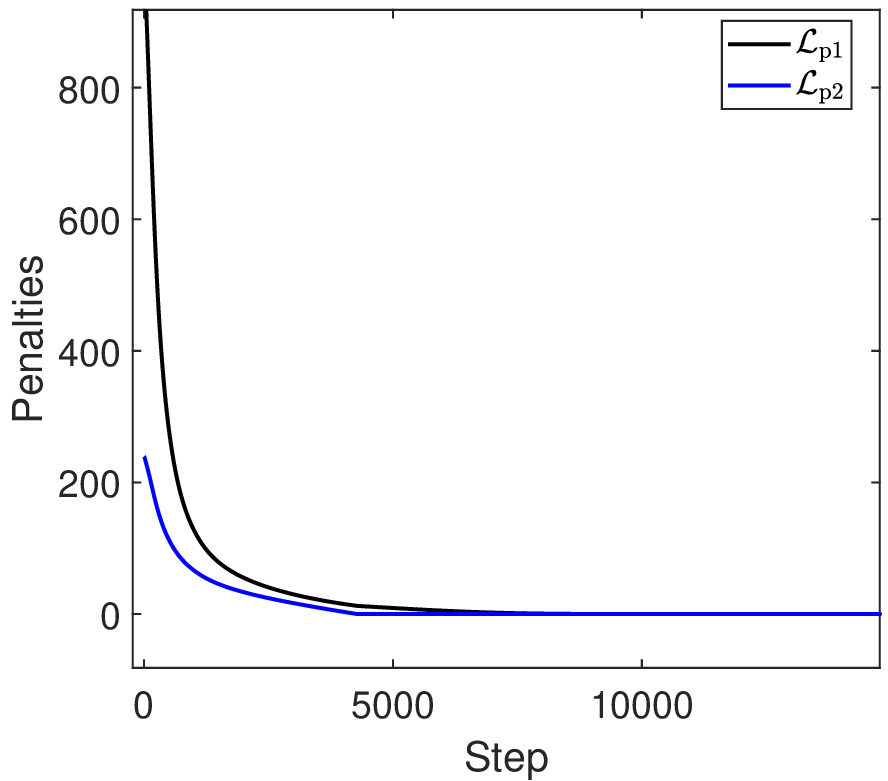}
        \end{matrix}
    \end{gather*}
\caption{Convergence behavior and penalty updates.}
\label{nn_part2}
\end{figure}

Here, we illustrate the convergence behavior of different neural network designs:
\begin{enumerate}
    \item Table \ref{convergenceTable}: The neural network value corresponds to the objective value $\int_{0}^{1} U(Q(s)) \, \mathrm{d}s$ in Problem \eqref{nn:goal}. The table also reports the minimum number of steps required to satisfy the SSD and budget constraints for each network design.

    \item Figure \ref{nn_part2}: Our algorithm-guided piecewise-neural-network framework reaches near-convergence in roughly 300 Adam steps, whereas a standard monolithic-neural-network requires approximately 10,000 steps to achieve comparable near-convergence. This dramatic acceleration highlights the efficiency of our approach. Furthermore, our piecewise framework achieves a higher sub-optimal objective value compared with the standard framework, indicating that the standard framework is prone to getting trapped in local optima, while our piecewise framework can avoid such traps and attain superior performance.

    \item Table \ref{convergenceTable} (i) vs (ii): The significant speedup under the piecewise framework primarily stems from rapid satisfaction of the penalty terms. By incorporating the analytic prior inspired from Algorithm \ref{alg} into the network design, our piecewise framework enforces the SSD and budget constraints more efficiently, allowing the neural network to converge quickly while maintaining better approximation of the sub-optimal solution.
\end{enumerate}

In summary, our study demonstrates the effectiveness of the algorithm-guided piecewise-neural-network framework for solving the SSD problem. For cases where an analytical optimal solution exists, our framework achieves results that closely match the optimal solution both numerically and structurally, indicating its high precision. From the perspective of convergence, our framework significantly accelerates the optimization process and is capable of escaping local optima, thereby attaining superior performance compared with the standard monolithic-neural-network framework. Our findings highlight efficiency and robustness of our approach across different model settings and suggest strong potential for extension to more complex scenarios.


\section{Conclusion}\label{sec:conc}
We study a utility maximization problem under stochastic dominance constraints. Starting from an S-shaped utility function, we derive the explicit optimal solution without a liquidation boundary under first-order stochastic dominance (FSD) constraints. 
For the more challenging SSD problem, particularly under non-concavity, obtaining an explicit optimal solution is inherently difficult. Motivated by the structural properties of the analytical theorem in the concave case, we introduce a Poor-Performance-Region Algorithm (PPRA). This algorithm efficiently identifies the candidate structure that a potential optimal solution should satisfy and proves effective in the vast majority of cases. Extensive numerical experiments illustrate how the algorithm operates, confirm its broad applicability, and demonstrate its capability to handle the few exceptional cases. Building on the structural insights provided by the PPRA, we further recognize the potential of neural networks in tackling SSD problems under non-concavity. While neural networks offer strong approximation capabilities, their direct application often suffers from slow convergence and severe sensitivity to local optima induced by non-concavity. By leveraging the PPRA’s ability to capture the essential structure of the optimal solution, we develop an algorithm-guided piecewise-neural-network framework that integrates these structural cues into the learning process. Compared with a standard neural-network approach, this framework achieves significantly faster convergence and effectively avoids undesirable local minima, thereby delivering consistently superior solution quality.

\subsubsection*{Acknowledgement} 
The authors are grateful to Jianming Xia and 
members of the research group on financial mathematics and risk management at The Chinese University of Hong Kong, Shenzhen for their insightful discussions and conversations. 
Y. Liu acknowledges financial support from the National Natural Science Foundation of China (Grant No. 12401624), The Chinese University of Hong Kong (Shenzhen) University Development Fund (Grant No. UDF01003336) and Shenzhen Science and Technology Program (Grant No. RCBS20231211090814028, JCYJ20250604141203005, 2025TC0010) and is partly supported by the Guangdong Provincial Key Laboratory of Mathematical Foundations for Artificial Intelligence (Grant No. 2023B1212010001). 


\begin{thebibliography}{99}
	\small
	\bibitem[Barberis and Thaler(2003)]{BT2003} Barberis, N., Thaler, R. (2003). {\it A Survey of Behavioral Finance}, in  {\it Handbook of the Economics of Finance}: \textbf{Vol. 1}. Financial Markets and Asset Pricing, M. H. G. M. Constantinides, and R. Stulz, eds., Elsevier, Kidlington, 1053-1128.

    
    \bibitem[Bichuch and Sturm(2014)]{BS2014}Bichuch, M., Sturm, S. (2014). Portfolio optimization under convex incentive schemes. \textit{Finance and Stochastics}, \textbf{18}, 873-915.
	
	\bibitem[Carpenter(2000)]{C2000} Carpenter, J. N. (2000). Does option compensation increase managerial risk appetite? \textit{Journal of Finance}, \textbf{55}, 2311-2331.

    \bibitem[Chen, Hieber and Nguyen(2019)]{CHN2019}  Chen, A., Hieber, P., Nguyen, T. (2019). Constrained non-concave utility maximization: an application to life insurance contracts with guarantees. \textit{European Journal of Operational Research}, \textbf{273}, 1119-1135.

    
    \bibitem[Dai et al.(2022)]{DKQW2022} Dai, M., Kou, S., Qian, S., Wan, X. (2022). Nonconcave utility maximization with portfolio bounds. \textit{Management Science}, \textbf{68}, 8368-8385.

    \bibitem[Dentcheva and Ruszczy\'nski(2003)]{DR2003} Dentcheva, D., Ruszczy\'nski, A. (2003). Optimization with stochastic dominance constraints. \textit{SIAM Journal on Optimization}, \textbf{14}(2), 548-566.

    \bibitem[Dentcheva and Ruszczy\'nski(2006)]{DR2006} Dentcheva, D., Ruszczy\'nski, A. (2006). Portfolio optimization with stochastic dominance constraints. \textit{Journal of Banking \& Finance}, \textbf{30}(2), 433-451.

    \bibitem[Dong and Zheng(2020)]{DZ2020} Dong, Y., Zheng, H. (2020). Optimal investment with S-shaped utility and trading and Value at Risk constraints: An application to defined contribution pension plan. \textit{European Journal of Operational Research}, \textbf{281}, 341-356.
    
	\bibitem[F\"ollmer and Schied(2016)]{FS2016} F\"ollmer, H., Schied, A. (2016). \emph{Stochastic Finance. An Introduction in Discrete Time}. Fourth Edition.  {Walter de Gruyter, Berlin}.

    \bibitem[Ghossoub and Zhu(2025)]{GZ2025} Ghossoub, M., Zhu, M. B. (2025). Risk-constrained portfolio choice under rank-dependent utility. \textit{Finance and Stochastics}, \textbf{29}, 399-442.
    
	\bibitem[{He and Kou(2018)}]{HK2018}
	He, X., Kou, S. (2018). Profit sharing in hedge funds. \textit{Mathematical Finance}, \textbf{28}, 50-81.
	
	\bibitem[He and Zhou(2011)]{HZ2011} He, X., Zhou, X. (2011). Portfolio choice under cumulative prospect theory: An analytical treatment. \textit{Management Science}, \textbf{57}, 315-331.
	
	
	
	\bibitem[Karatzas et al.(1987)]{KLS1987} Karatzas, I., Lehoczky, J. P., Shreve, S. E. (1987). Optimal portfolio and consumption decisions for a ``small investor" on a finite horizon. \textit{SIAM Journal on Control and Optimization}, \textbf{25}, 1557-1586.
	
	
	
	
	
	\bibitem[Kahneman and Tversky(1979)]{KT1979} Kahneman, D., Tversky, A. (1979). Prospect Theory: an analysis of decision under risk. \textit{Econometrica}, \textbf{47}, 263-291.
	
	\bibitem[Liang and Liu(2020)]{LL2020} Liang, Z., Liu, Y. (2020).  A classification approach to general S-shaped utility optimization with principals' constraints.
	\textit{SIAM Journal on Control and Optimization}, \textbf{58}, 3734-3762.

    \bibitem[Liang and Liu(2024)]{LL2024} Liang, Z., Liu, Y. (2024). An asymptotic approach to centrally planned portfolio selection. \textit{Advances in Applied Probability}, \textbf{56}, 757-784.

    \bibitem[Liang, Liu and Zhang(2025)]{LLZ2025} Liang, Z., Liu, Y., Zhang, L. (2025). A framework of state-dependent utility optimization with general benchmarks. \textit{Finance and Stochastics}, \textbf{29}, 469-518.
    
	\bibitem[Merton(1969)]{M1969} Merton, R. C. (1969). Lifetime portfolio selection under uncertainty: The continuous-time case. \textit{Review of Economics and Statistics}, \textbf{51}, 247-257.

    \bibitem[Pesenti, Wang and Wang(2025)]{PWW2025} Pesenti, S., Wang, Q., Wang, R. (2025).  
	Optimizing distortion riskmetrics with distributional uncertainty. \textit{Mathematical Programming}, \textbf{213}, 51-106.

    \bibitem[Rockafellar(1970)]{R1970} Rockafellar, R. T. (1970). \textit{Convex Analysis}. Princeton University Press, 1st edition.


    \bibitem[Scarselli and Tsoi(1998)]{ST1998} Scarselli, F., Tsoi, A. C. (1998). Universal approximation using feedforward neural networks: A survey of some existing methods, and some new results. \textit{Neural Networks}, \textbf{11}(1), 15-37.
    
	\bibitem[Tversky and Kahneman(1992)]{TK1992}Tversky, A., Kahneman, D. (1992). Advances in prospect theory: cumulative representation of uncertainty. \textit{Journal of Risk and Uncertainty}, \textbf{5}, 297-323.


    
    \bibitem[Wang and Xia(2021)]{WX2021} Wang, X., Xia, J. (2021). Expected utility maximization with stochastic dominance constraints in complete markets. \textit{SIAM Journal on Financial Mathematics}, \textbf{12}, 1054-1111.

    \bibitem[Wang, Wei and Xia(2024)]{WWX2024} Wang, Y., Wei, J., Xia, J. (2024).  Mean-Stochastic-Dominance portfolio selection in continuous time. \textit{SIAM Journal on Financial Mathematics}, \textbf{15}, SC80-SC90. 
    
	
	\bibitem[Wei(2018)]{W2018} Wei, P. (2018). Risk management with weighted VaR. \textit{Mathematical Finance}, \textbf{28}, 1020-1060.

    \bibitem[Xia and Zhou(2016)]{XZ2016} Xia, J., Zhou, X. (2016). Arrow-Debreu equilibria for rank-dependent utilities. \textit{Mathematical Finance}, \textbf{26}, 558-588.
	
	\bibitem[Xu(2016)]{X2016} Xu, Z. (2016). A note on the quantile formulation. \textit{Mathematical Finance}, \textbf{26}, 589-601.



    

\end{thebibliography}
\end{document}